\documentclass[a4paper, USenglish, cleveref, autoref, thm-restate]{lipics-v2021}
\usepackage[T1]{fontenc}
\usepackage[utf8]{inputenc}
\usepackage{listings}
\usepackage{cite}

\usepackage{framed}
\usepackage{doc}
\usepackage{mathtools}
\usepackage[customcolors]{hf-tikz}

\usepackage{float}
\usepackage[misc]{ifsym}
\usepackage[ruled]{algorithm2e}
\usepackage{mathtools}
\usepackage{amsthm}
\usepackage{amsfonts}
\usepackage{amsmath}
\usepackage{amssymb}
\usepackage{todonotes}
\usepackage{nicefrac}
\usepackage{booktabs}
\usepackage{multirow}
\usepackage{xpatch}
\usepackage{appendix}
\usepackage{adjustbox}
\usepackage{interval}
\usepackage{microtype}
\usepackage{subcaption}

\usepackage{tikz}
\usetikzlibrary{shapes.geometric}
\usetikzlibrary{fit}
\usetikzlibrary{arrows}
\usetikzlibrary{patterns}
\usetikzlibrary{matrix}
\usetikzlibrary{positioning}
\usetikzlibrary{backgrounds}
\usetikzlibrary{angles,quotes}
\usetikzlibrary{patterns.meta}

\hypersetup{
  colorlinks = true,
  citecolor = blue,
  linkcolor = blue,
  urlcolor = blue
}

\newcommand{\orvar}{\mathsf{o}}
\newcommand{\uvar}{\mathsf{u}^4}
\newcommand{\ufvar}{\mathsf{u}^5}
\newcommand{\vvar}{\mathsf{v}^4}
\newcommand{\hvar}{\mathsf{h}}
\newcommand{\cvar}{\mathsf{c}}
\newcommand{\ov}[1]{\overline{#1}}

\usepackage{color}
\definecolor{keywordcolor}{rgb}{0.7, 0.1, 0.1}   % red
\definecolor{tacticcolor}{rgb}{0.0, 0.1, 0.6}    % blue
\definecolor{commentcolor}{rgb}{0.4, 0.4, 0.4}   % grey
\definecolor{symbolcolor}{rgb}{0.0, 0.1, 0.6}    % blue
\definecolor{sortcolor}{rgb}{0.1, 0.5, 0.1}      % green
\definecolor{attributecolor}{rgb}{0.7, 0.1, 0.1} % red

\makeatletter
\def\orcidID#1{\href{http://orcid.org/#1}{\protect\raisebox{-1.25pt}{\protect\includegraphics{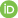}}}}
\makeatother

\newtheorem*{theorem*}{Theorem}

\title{\texorpdfstring{Formal Verification of the Empty Hexagon Number}{A Formal Verification of the Empty Hexagon Number}}%
\author{Bernardo {Subercaseaux}}{Carnegie Mellon University}{bsuberca@andrew.cmu.edu}{https://orcid.org/0000-0003-2295-1299}{}
\author{Wojciech {Nawrocki}}{Carnegie Mellon University}{wjnawrocki@cmu.edu}{https://orcid.org/0000-0002-8839-0618}{}
\author{James {Gallicchio}}{Carnegie Mellon University}{jgallicc@andrew.cmu.edu}{https://orcid.org/0000-0002-0838-3240}{}
\author{Cayden {Codel}}{Carnegie Mellon University}{ccodel@andrew.cmu.edu}{https://orcid.org/0000-0003-3588-4873}{}
\author{Mario {Carneiro}}{Carnegie Mellon University}{mcarneir@andrew.cmu.edu}{https://orcid.org/0000-0002-0470-5249}{}
\author{Marijn J. H. {Heule}}{Carnegie Mellon University}{mheule@andrew.cmu.edu}{https://orcid.org/0000-0002-5587-8801}{}

\nolinenumbers
\authorrunning{Subercaseaux et al.}

\titlerunning{Formal Verification of the Empty Hexagon Number}

\ccsdesc[300]{Theory of computation~Logic and verification}

\keywords{Empty Hexagon Number, Discrete Computational Geometry, Erd\H{o}s-Szekeres}

\Copyright{The authors}%Bernardo Subercaseaux and Wojciech Nawrocki and James Gallicchio and Cayden Codel and Mario Carneiro}

\begin{document}

\maketitle

\begin{abstract}
  A recent breakthrough in computer-assisted mathematics showed that every set of $30$ points in the plane in general position (i.e., no three points on a common line) contains an empty convex hexagon. %, closing a line of research dating back to the 1930s.
  With a combination of geometric insights and automated reasoning techniques, Heule and Scheucher constructed CNF formulas $\phi_n$, with $O(n^4)$ clauses, such that if $\phi_n$ is unsatisfiable then every set of $n$ points in general position must contain an empty convex hexagon.
  An unsatisfiability proof for $n = 30$ was then found with a SAT solver using 17\,300 CPU hours of parallel computation. %, thus implying that the empty hexagon number $h(6)$ is equal to 30.
  In this paper, we formalize and verify this result in the Lean theorem prover. Our formalization covers ideas in discrete computational geometry and SAT encoding techniques by introducing a framework that connects geometric objects to propositional assignments.
  %  that have been successfully applied to similar Erd\H{o}s-Szekeres-type problems.
  % In particular, our framework connects standard geometric objects to propositional assignments.
  We see this as a key step towards the formal verification of other SAT-based results in geometry, since the abstractions we use have been successfully applied to similar Erd\H{o}s-Szekeres-type problems.
  Overall, we hope that this work sets a new standard for verification when extensive computation is used for discrete geometry problems, and that it increases the trust the mathematical community has in computer-assisted proofs in this area.
\end{abstract}

\section{Introduction}\label{sec:intro}
Mathematicians are often rightfully skeptical of proofs that rely on extensive computation (e.g., the controversy around the four color theorem~\cite{Walters2004ItAT}).
Nonetheless, many mathematically-interesting theorems have been resolved that way.
SAT solving in particular has been a powerful tool for mathematics, successfully resolving
Keller's conjecture~\cite{brakensiek2023resolution},
the packing chromatic number of the infinite grid~\cite{Subercaseaux_Heule_2023},
the Pythagorean triples problem~\cite{Heule_2016},
Lam's problem~\cite{21bright_sat_based_resolution_lams_problem},
and one case of the Erd\H{o}s discrepancy conjecture~\cite{konev2014sat}.
All of these proofs rely on the same two-step structure:
\begin{itemize}
\item \textbf{(Reduction)} Show that the mathematical theorem of interest is true if a concrete propositional formula~$\phi$ is unsatisfiable.
\item \textbf{(Solving)} Show that $\phi$ is indeed unsatisfiable.
\end{itemize}

% \footnote{A variant of this procedure uses \emph{satisfying assignments} for $P$ to construct explicit witnesses for the original theorem, but we focus on the unsatisfiable case.}
% \end{itemize}

Formal methods researchers have devoted significant attention to making the \emph{solving} step reliable, reproducible and trustworthy.
Modern SAT solvers produce proofs of unsatisfiability in formal systems
such as DRAT~\cite{drat-trim14}
that can in turn be checked with verified proof checkers
such as \texttt{cake\_lpr}~\cite{tanVerifiedPropagationRedundancy2023}. 
These tools ensure that when a SAT solver declares a formula~$\phi$ to be unsatisfiable, the formula is indeed unsatisfiable.
In contrast, the \emph{reduction} step can use problem-specific mathematical insights that, when left unverified, threaten the trustworthiness of SAT-based proofs in mathematics. 
% The correctness of the \emph{reduction} step, however, has received 
% leaving room for doubt that any results relying on such reduction arguments are correct.
A perfect example of the complexity of this reduction step can be found in a recent breakthrough of Heule and Scheucher~\cite{emptyHexagonNumber} in discrete computational geometry. 
They constructed (and solved) a formula $\phi$ whose unsatisfiability implies that every set of 30 points, without three in a common line, must contain an empty convex hexagon.
However, as is common with such results, their reduction argument was only sketched, relied heavily on intuition,
and left several gaps to be filled in.

% result from Heule and Scheucher~\cite{emptyHexagonNumber} falls into this category.
% They resolve a variant of the Happy Ending Problem,
% in particular that every set of 30 points in general position contains an empty convex hexagon.
% Their proof relies on a complicated reduction to SAT
% involving numerous nontrivial geometric optimizations and symmetry-breaking arguments.

In this paper we complete and formalize the reduction of Heule and Scheucher in the Lean theorem prover~\cite{demouraLeanTheoremProver2015}. We do so by connecting existing geometric definitions
in the mathematical proof library \textsf{mathlib}~\cite{The_mathlib_Community_2020}
to the unsatisfiability of a particular SAT instance, thus setting a new standard for verifying results which rely on extensive computation.
Our formalization is publicly available at \url{https://github.com/bsubercaseaux/EmptyHexagonLean/tree/itp2024}.

\subparagraph*{Verification of SAT proofs.}
Formal verification plays a crucial role in certifying the \emph{solving} step of SAT-based results.
For example, theorem provers and formal methods tools have been used to verify solvers~\cite{10maric_formal_verification_modern_sat_solver_shallow_embedding_isabelle_hol,oeVersatVerifiedModern2012,skotam_creusat_2022}
and proof checkers~\cite{lammichEfficientVerifiedSAT2020,tanVerifiedPropagationRedundancy2023}.
However, the \emph{reduction} step has not received similar scrutiny.
Some work has been done to verify the reductions to SAT underlying these kinds of mathematical results.
The solution to the Pythagorean triples problem
was verified in the \textsf{Coq} proof assistant
by Cruz-Filipe and coauthors~\cite{formalPythagoreanTriples,LPAR-21:Formally_Proving_Boolean_Pythagorean}.
More generally,
Giljeg\r{a}rd and Wennerbreck~\cite{GilAndWennerbeck} provide a \textsf{CakeML} library
of verified SAT encodings,
which they used to write verified reductions from different puzzles
(e.g., Sudoku, Kakuro, the \emph{N-queens} problem).
The reduction verification techniques we use in this paper
are based on that of Codel, Avigad, and Heule~\cite{Cayden} in the Lean theorem prover.

Formal verification for SAT-based combinatorial geometry
was pioneered by Marić~\cite{19maric_fast_formal_proof_erdos_szekeres_conjecture_convex_polygons_most_six_points}.
He developed a reduction of a case of the Happy Ending Problem to SAT
and formally verified it in \textsf{Isabelle/HOL}.
We give a detailed comparison between his work and ours in~\Cref{sec:related-work}.

\subparagraph*{Lean.}
Initially developed by Leonardo de Moura in 2013~\cite{demouraLeanTheoremProver2015},
the Lean theorem prover has become a popular choice for formalizing modern mathematical research.
Recent successes include the~\emph{Liquid Tensor Experiment}~\cite{Castelvecchi2021}
and the proof of the polynomial Freiman-Ruzsa conjecture~\cite{gowers2023conjecture, slomanATeamMathProves2023},
both of which brought significant attention to Lean.
A major selling point for Lean is the \textsf{mathlib} project~\cite{The_mathlib_Community_2020},
a monolithic formalization of foundational mathematics.
By relying on \textsf{mathlib} for definitions, lemmas, and proof tactics,
mathematicians can focus on the interesting components of a formalization
while avoiding duplication of proof efforts across formalizations.
In turn, by making a formalization compatible with \textsf{mathlib},
future proof efforts can rely on work done today.
In this spirit, we connect our results to~\textsf{mathlib} as much as possible.

\subparagraph*{The Empty Hexagon Number.}
In the 1930s,
Erd\H{o}s and Szekeres, inspired by Esther Klein, showed that for any $k \geq 3$,
one can find a sufficiently large number $n$
such that every $n$ points in \emph{general position}
(i.e., with no three points collinear)
contain a convex \emph{$k$-gon}, i.e., a convex polygon with $k$ vertices~\cite{35erdos_combinatorial_problem_geometry}.
The minimal such $n$ is denoted $g(k)$.
The same authors later showed that $g(k) > 2^{k-2}$
and conjectured that this bound is tight~\cite{60erdos_some_extremum_problems_elementary_geometry}.
Indeed, it is known that $g(5) = 9$ and $g(6) = 17$,
with the latter result obtained by Szekeres and Peters 71 years after the initial conjecture
via exhaustive computer search~\cite{06szekeres_computer_solution_17_point_erdos_szekeres_problem}.
Larger cases remain open,
with $g(k) \leq 2^{k+o(k)}$ the best known upper bound~\cite{suk2017erdos,holmsen2017two}.
This problem is now known as the \emph{Happy Ending Problem},
as it led to the marriage of Klein and Szekeres.

In a similar spirit,
Erd\H{o}s defined $h(k)$
to be the minimal number of points in general position
that is guaranteed to contain a \emph{$k$-hole},
or \emph{empty $k$-gon},
meaning a convex $k$-gon with no other point inside.
It is easy to check that $h(3) = 3$ and $h(4) = 5$.
In 1978, Harborth established that $h(5) = 10$~\cite{Harborth1978}.
Surprisingly, in 1983, Horton discovered constructions of arbitrarily large point sets that 
avoid $k$-holes for $k \geq 7$~\cite{hortonSetsNoEmpty1983}.
Only $h (6)$ remained.
The \emph{Empty Hexagon Theorem},
establishing $h(6)$ to be finite,
was proven independently by Gerken and Nicolás in 2006~\cite{gerkenEmptyConvexHexagons2008,nicolasEmptyHexagonTheorem2007}.
In 2008, Valtr narrowed the range of possible values down to $30 \leq h(6) \leq 1717$,
where the problem remained until the breakthrough by Heule and Scheucher~\cite{emptyHexagonNumber},
who used a SAT solver to prove that $h(6) \leq 30$,
a result we refer to as the \emph{Empty Hexagon Number}.

\section{Outline of the proof}\label{sec:outline}
We will incrementally build sufficient machinery to prove:

\begin{theorem*}
Any finite set of $30$ or more points in the plane in general position has a $6$-hole.
\end{theorem*}

\begin{proof}[Outline of the proof]
We begin \Cref{sec:geometry} with a precise statement in Lean of the above theorem and involved geometric terms.
In a nutshell, the proof consists of building a CNF formula~$\phi_n$ such that 
from any set $S$ of $n$ points in general position without a $6$-hole we can construct a satisfying assignment $\tau_S$ for $\phi_n$.
Then, checking that $\phi_{30}$ is unsatisfiable implies that no such set $S$ of size $30$ exists, thus implying the theorem. 
In order to construct $\phi_n$, one must first discretize the continuous space $\mathbb{R}^2$. \emph{Triple orientations}, presented in~\Cref{sec:triple-orientations}, are a way to achieve this. Concretely, any three points $p,q,r$ in general position correspond to either a clockwise turn, denoted by $\sigma(p, q, r) = -1$, or a counterclockwise turn, denoted by $\sigma(p, q, r) = +1$, depending on whether $r$ is above the directed line $\overrightarrow{pq}$ or not. 
In this way, every set $S$ of points in general position induces an assignment $\sigma_S: S^3 \to \{-1,+1\}$
of triple orientations.
We show in~\Cref{sec:triple-orientations} that whether $S$ contains a $k$-hole (i.e., \lstinline|HasEmptyKGon k S|) depends entirely on $\sigma_S$. 
As each orientation $\sigma(p, q, r)$ can only take two values, we can represent each orientation $\sigma(p,q,r)$ with a boolean variable. Any set of points $S$ in general position thus induces an assignment $\tau_S$ over its \emph{orientation variables}. 
Because \lstinline|HasEmptyKGon k S| depends only on $\sigma_S$, it can be written as a boolean formula over the orientation variables.
% \[ 
%     \exists S : \text{set of points in general position without a } k\text{-hole} \implies   
% \]
% Therefore, $\phi_n$ can be constructed by writting the statement 
% \[
% \exists \sigma_S : S^3 \to \{-1,+1\}, \text{ such that } \neg \text{\lstinline|HasEmpty6gon S|}[\sigma_S]    
% \]
% 
Unfortunately, it is practically infeasible to determine if such a formula is satisfiable with a naïve encoding.
In order to create a better encoding,~\Cref{sec:symmetry-breaking} shows that one can assume, without loss of generality, that the set of points $S$ is in \emph{canonical position}.
Canonicity eliminates a number of symmetries from the problem --
ordering, rotation, and mirroring --
significantly reducing the search space.
In~\Cref{sec:encoding}, we show the correctness of the efficient encoding of Heule and Scheucher~\cite{emptyHexagonNumber} for constructing
a smaller CNF formula $\phi_n$.
Concretely, we show that any finite set of $n$ points in canonical position
containing no $6$-hole
would give rise to a propositional assignment $\tau_S$ satisfying $\phi_n$.
However, $\phi_{30}$ (depicted in~\Cref{fig:full-encoding}) is unsatisfiable;
therefore no such set of size $30$ exists
and the theorem follows by contradiction.
As detailed in~\Cref{sec:encoding},
to establish unsatisfiability of $\phi_{30}$
we passed the formula produced by our verified encoder to a SAT solver,
and used a verified proof checker to certify the correctness of
the resulting unsatisfiability proof.
The construction of $\phi_n$ and $\tau_S$
involves sophisticated optimizations
which we justify using geometric arguments.
\end{proof}
% file-local attic:

%We compare our work with its closest precedent due to Marić~\cite{19maric_fast_formal_proof_erdos_szekeres_conjecture_convex_polygons_most_six_points}
%in~\Cref{sec:related-work}.
%We conclude in~\Cref{sec:conclusions} by discussing next steps
%towards the formal verification of other Erd\H{o}s-Szekeres-type problems.

% Then,~\Cref{sec:empty-triangle} presents how the previous elements are already enough
% to formalize a SAT-based proof for the \emph{Empty Triangle Theorem},
% a much simpler problem involving only triangles.

% Namely, that one can assume without loss of generality the following two properties at the same time: (i) points are labeled from left to right without two of them having the same $x$-coordinate, and (ii) the triples $(p_1, p_i, p_j)$ are always oriented counterclockwise for $i < j$.

\section{Geometric Preliminaries}\label{sec:geometry}
We identify points with elements of $\mathbb{R}^2$. Concretely,~\lstinline|abbrev Point := EuclideanSpace ℝ (Fin 2)|.
The next step is to define what it means for a $k$-gon to be \emph{empty} (with respect to a set of points) and \emph{convex}, which we do in terms of \textsf{mathlib} primitives.

\begin{lstlisting}
/-- `EmptyShapeIn S P' means that `S' carves out an empty shape in `P':
the convex hull of `S' contains no point of `P' other than those already in `S'. -/
def EmptyShapeIn (S P : Set Point) : Prop :=
    ∀ p ∈ P \ S, p ∉ convexHull ℝ S

/-- `ConvexPoints S' means that `S' consists of extremal points of its convex hull,
i.e., the point set encloses a convex polygon. -/
def ConvexPoints (S : Set Point) : Prop :=
    ∀ a ∈ S, a ∉ convexHull ℝ (S \ {a})

def ConvexEmptyIn (S P : Set Point) : Prop :=
    ConvexPoints S ∧ EmptyShapeIn S P

def HasEmptyKGon (k : Nat) (S : Set Point) : Prop :=
    ∃ s : Finset Point, s.card = k ∧ ↑s ⊆ S ∧ ConvexEmptyIn s S
\end{lstlisting}

Let \lstinline|ListInGenPos| be a predicate that states that a list of points is in \emph{general position}, i.e., no three points lie on a common line (made precise in~\Cref{sec:triple-orientations}).
With this we can already state the main theorem of our paper.

\begin{lstlisting}
theorem hole_6_theorem (pts : List Point) (gp : ListInGenPos pts)
    (h : pts.length ≥ 30) : HasEmptyKGon 6 pts.toFinset
\end{lstlisting}

At the root  of the encoding of Heule and Scheucher is the idea that the~\lstinline|ConvexEmptyIn| predicate can be determined by analyzing only triangles. In particular, that a set $s$ of $k$ points in a pointset $S$ form an empty convex $k$-gon if and only if all the ${k \choose 3}$ triangles induced by vertices in $s$ are empty with respect to $S$. This is discussed informally in~\cite[Section 3, Eq.~4]{emptyHexagonNumber}.
Concretely, we prove the following theorem:
\begin{lstlisting}
theorem ConvexEmptyIn.iff_triangles {s : Finset Point} {S : Set Point}
    (sS : ↑s ⊆ S) (sz : 3 ≤ s.card) :
    ConvexEmptyIn s S ↔
    ∀ (t : Finset Point), t.card = 3 → t ⊆ s → ConvexEmptyIn t S
\end{lstlisting}
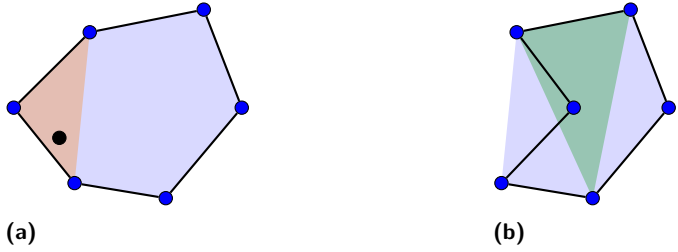
\begin{figure}
    \centering
    \begin{subfigure}{0.45\textwidth}
        \begin{tikzpicture}
        \coordinate (a) at (0,0);
        \coordinate (b) at (1, 1);
        \coordinate (c) at (2.5, 1.3);
        \coordinate (d) at (3, 0);
        \coordinate (e) at (0.8, -1);
        \coordinate (f) at (2, -1.2);

        \fill[blue, opacity=0.15] (a) -- (b) -- (c) -- (d) -- (f) -- (e) -- cycle;
        \fill[orange, opacity=0.3] (a) -- (b) -- (e) -- (a) -- cycle;
        % \fill[blue, opacity=0.3] (b) -- (f) -- (e) -- (b) -- cycle;
        % \fill[green!60!black, opacity=0.3] (b) -- (c) -- (f) -- (b) -- cycle;
        % \fill[yellow!60!black, opacity=0.3] (c) -- (d) -- (f) -- (c) -- cycle;

        \node[draw, circle, black, fill=blue, inner sep=0pt, minimum size=5pt] (pA) at (a) {};
        \node[draw, circle, black, fill=blue, inner sep=0pt, minimum size=5pt] (pB) at (b) {};
        \node[draw, circle, black, fill=blue, inner sep=0pt, minimum size=5pt] (pC) at (c) {};
        \node[draw, circle, black, fill=blue, inner sep=0pt, minimum size=5pt] (pD) at (d) {};
        \node[draw, circle, black, fill=blue, inner sep=0pt, minimum size=5pt] (pE) at (e) {};
        \node[draw, circle, black, fill=blue, inner sep=0pt, minimum size=5pt] (pF) at (f) {};

        \node[draw, circle, black, fill=black, inner sep=0pt, minimum size=5pt] (pG) at (0.6, -0.4) {};

        \draw[thick] (pA) -- (pB) -- (pC) -- (pD) -- (pF) -- (pE) -- (pA);
        % \draw[dashed, thick, red] (pA) -- (pB) -- (pE) -- (pA);
        % \draw[dashed, thick, blue] (pB) -- (pF) -- (pE) -- (pB);
        % \draw[dashed, thick, green!60!black] (pB) -- (pC) -- (pF) -- (pB);

        \end{tikzpicture}
        \caption{}\label{fig:triangulation-a}
    \end{subfigure}
    \begin{subfigure}{0.45\textwidth}
        \begin{tikzpicture}
            \coordinate (a) at (1.75,0);
            \coordinate (b) at (1, 1);
            \coordinate (c) at (2.5, 1.3);
            \coordinate (d) at (3, 0);
            \coordinate (e) at (0.8, -1);
            \coordinate (f) at (2, -1.2);
            \fill[blue, opacity=0.15]  (b) -- (c) -- (d) -- (f) -- (e) -- (b) -- cycle;
            % \fill[orange, opacity=0.3] (a) -- (b) -- (e) -- (a) -- cycle;
            % \fill[blue, opacity=0.3] (b) -- (f) -- (e) -- (b) -- cycle;
            \fill[green!60!black, opacity=0.3] (b) -- (c) -- (f) -- (b) -- cycle;
            % \fill[yellow!60!black, opacity=0.3] (c) -- (d) -- (f) -- (c) -- cycle;

            \node[draw, circle, black, fill=blue, inner sep=0pt, minimum size=5pt] (pA) at (a) {};
            \node[draw, circle, black, fill=blue, inner sep=0pt, minimum size=5pt] (pB) at (b) {};
            \node[draw, circle, black, fill=blue, inner sep=0pt, minimum size=5pt] (pC) at (c) {};
            \node[draw, circle, black, fill=blue, inner sep=0pt, minimum size=5pt] (pD) at (d) {};
            \node[draw, circle, black, fill=blue, inner sep=0pt, minimum size=5pt] (pE) at (e) {};
            \node[draw, circle, black, fill=blue, inner sep=0pt, minimum size=5pt] (pF) at (f) {};

            % \node[draw, circle, black, fill=black, inner sep=0pt, minimum size=5pt] (pG) at (0.6, -0.4) {};

            \draw[thick] (pA) -- (pB) -- (pC) -- (pD) -- (pF) -- (pE) -- (pA);
            % \draw[dashed, thick, red] (pA) -- (pB) -- (pE) -- (pA);
            % \draw[dashed, thick, blue] (pB) -- (pF) -- (pE) -- (pB);
            % \draw[dashed, thick, green!60!black] (pB) -- (pC) -- (pF) -- (pB);

        \end{tikzpicture}
        \caption{}\label{fig:triangulation-b}
    \end{subfigure}
    \caption{Illustration of the proof for \lstinline|ConvexEmptyIn.iff_triangles|. The left subfigure shows how a point in $S \setminus s$ that lies inside  $s$ will be inside one of the triangles induced by the convex hull of $s$ (orange triangle). The right subfigure shows how if the $\textsf{ConvexPoints}$ predicate does not hold of $s$, then some point $a \in s$ will be inside one of the triangles induced by the convex hull of $s \setminus \{a\}$.}\label{fig:triangulation}
\end{figure}
\begin{proof}[Proof sketch]
    We first prove a simple monotonicity lemma: if $\textsf{ConvexPoints}(s)$, then $\textsf{ConvexPoints}(s')$ for every $s' \subseteq s$, and similarly $\textsf{EmptyShapeIn}(s, S) \Rightarrow \textsf{EmptyShapeIn}(s', S)$ for every set of points $S$.
    By instantiating this monotonicity lemma over all subsets $t \subseteq s$ with $|t| = 3$ we get the forward direction of the theorem.
    For the backward direction it is easier to reason contrapositively: if the~$\textsf{ConvexPoints}$ predicate does not hold of $s$, or if $s$ is not empty w.r.t.~$S$, then we want to show that there is a triangle $t \subseteq s$ that is also not empty w.r.t.~$S$. To see this, let $H$ be the convex hull of $s$, and then by Carath\'eodory's theorem (cf. \lstinline|theorem convexHull_eq_union| from \textsf{mathlib}), every point in $H$ is a convex combination of at most $3$ points in $s$, and consequently, of exactly $3$ points in $s$.
    If $s$ is non-empty w.r.t.~$S$, then there is a point $p \in S \setminus s$ that belongs to $H$, and by Carath\'eodory, $p$ is a convex combination of $3$ points in $s \setminus \{a\}$, and thus lies inside a triangle $t \subseteq s$ (\Cref{fig:triangulation-a}). If $s$ does not hold $\textsf{ConvexPoints}$, then there is a point $a \in s$ such that $a \in \textsf{convexHull}(s \setminus \{a\})$, and by Carath\'eodory again, $a$ is a convex combination of $3$ points in $s$, and thus lies inside a triangle $t \subseteq s \setminus \{a\}$ (\Cref{fig:triangulation-b}).
%     The
%     The forward direction is easy to see, as triangles are always convex, and if $s$ is empty w.r.t.~$S$, then so are the triangles induced by vertices of $s$.
%     Let $T = \{t_1, \ldots, t_{{k \choose 3}}\}$ be all the triangles induced by vertices of $s$.
%    For the backward direction we will need a \emph{triangulation lemma} stating that the convex hull of $s$ can be partitioned into triangles $P = \{t_i, \ldots, t_j\}$, and $P \subseteq T$.
%     To see that if every $t \in T$ is empty w.r.t. $S$ then $s$ is also empty w.r.t. $S$ we can use the contrapositive statement.
    %  Indeed, assume $s$ is not empty w.r.t. $S$, there is a point $p \in S \setminus s$ that lies inside the convex hull of $s$. Because $P$ is a partition of the convex hull of $s$, point $p$ must be inside some $t \in P \subseteq T$.
    %  To see convexity, we can reason contrapositively again. If $s$ is not convex, then there is a point $p \in s$ that lies inside the convex hull of $s$, and thus lies inside a triangle $t \in P$.
\end{proof}
    
The next section shows how boolean variables can be used to encode which triangles are empty w.r.t.~a pointset, which as the previous theorem shows, can be used to encode the presence or absence of $k$-holes.

\section{Triple Orientations}\label{sec:triple-orientations}
An essential step for obtaining computational proofs of geometric results is \emph{discretization}: problems concerning the existence of an object $\mathcal{O}$ in a continuous space such as $\mathbb{R}^2$ must be reformulated in terms of the existence of a discrete and finitely representable object $\mathcal{O}'$ that a computer can find (or discard its existence).
This poses a particular challenge for problems in which the desired geometric object $\mathcal{O}$ is characterized by very specific coordinates of points, requiring to deal with floating point arithmetic or numerical instability.
Fortunately, this is not the case for Erd\H{o}s-Szekeres-type problems such as determining the value of $h(k)$, which are naturally well-suited for computation.
This is so because the properties of interest (e.g., convexity, emptiness) can be described in terms of axiomatizable relationships between points and lines (e.g., point $p$ is above the line $\overrightarrow{qr}$, lines $\overrightarrow{qr}$ and $\overrightarrow{st}$ intersect, etc.), which are invariant under rotations, translations, and even small perturbations of the coordinates. This suggests the problems can be discretized in terms of boolean variables representing these relationships, forgetting the specific coordinates of the points.
The combinatorial abstraction that has been most widely used in Erd\H{o}s-Szekeres-type problems is that of \emph{triple orientations}~\cite{ emptyHexagonNumber, scheucherTwoDisjoint5holes2020}. This concept is also known as \emph{signotopes}~\cite{felsnerSweepsArrangementsSignotopes2001,subercaseaux2023minimizing},  Knuth's \emph{counterclockwise} relation~\cite{knuthAxiomsHulls1992}, or \emph{signatures}~\cite{szekeres_peters_2006}.
Given points $p, q, r$, their \emph{triple-orientation} is defined as
\newcommand{\sign}{\operatorname{sign}}
\[
  \sigma(p, q, r) = \sign \det \begin{pmatrix} p_x & q_x & r_x \\ p_y & q_y & r_y \\ 1 & 1 & 1 \end{pmatrix} = \begin{cases}
    1 & \text{if } p, q, r \text{ are \emph{oriented} counterclockwise}, \\
    0 & \text{if } p, q, r \text{ are collinear}, \\
    -1 & \text{if } p, q, r \text{ are \emph{oriented} clockwise}.
  \end{cases}.
\]

\begin{figure}
  \centering
\begin{tikzpicture}
  %\draw[ultra thick, dashed, blue] (5,1) -- (0,0);
  \node[draw, circle, black, fill=blue, inner sep=0pt, minimum size=5pt] (p) at (0,0) {};
  \node[] at (-0.2, 0.25) {$p$};
  \node[draw, circle, black, fill=blue, inner sep=0pt, minimum size=5pt] (q) at (5,1) {};
  \node[] at (5.2, 0.75) {$q$};
  \node[draw, circle, black, fill=blue, inner sep=0pt, minimum size=5pt] (r) at (2,3) {};
  \node[] at (2, 3.25) {$r$};

  \node[draw, circle, black, fill=blue, inner sep=0pt, minimum size=5pt] (s) at (1.5, 1) {};
  \node[] at (1.35, 1.2) {$s$};

  \node[draw, circle, black, fill=blue, inner sep=0pt, minimum size=5pt] (t) at (4.5, 3) {};
  \node[] at (4.3, 3.25) {$t$};

  \draw[ thick,  green!60!black] (p) -- (r);
  \draw[ thick,  green!60!black, -latex] (r) -- (q);

  \draw[ thick,  black] (r) -- (s);
  \draw[ thick,  black, -latex] (s) -- (q);

  \draw[thick,  red] (p) -- (s);
  \draw[thick,  red, -latex] (s) -- (t);
  % \draw[fill=green, opacity=0.5] (a.center) -- (b.center) -- (c.center) -- cycle;
\end{tikzpicture}
\caption{Illustration of triple orientations, where $\sigma(p, r, q) = -1, \sigma(r, s, q) = 1, $ and $\sigma(p, s, t) = 0$.}\label{fig:triple-orientation}
\end{figure}
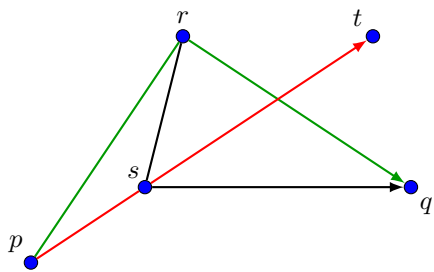

An example is illustrated in~\Cref{fig:triple-orientation}. We directly use \textsf{mathlib}'s definition of the determinant to define $\sigma$.
% @[pp_dot] abbrev x (p : Point) : ℝ := p 0
% @[pp_dot] abbrev y (p : Point) : ℝ := p 1
\begin{lstlisting}
inductive Orientation : Type where
  | cw -- clockwise :=  -
  | ccw -- counter clockwise := +
  | collinear -- := 0

noncomputable def σ (p q r : Point) : Orientation :=
  let det := Matrix.det !![p.x, q.x, r.x ; p.y, q.y, r.y ; 1, 1, 1]
  if 0 < det then ccw
  else if det < 0 then cw
  else collinear
\end{lstlisting}

% def Orientation.ofReal (r : ℝ) : Orientation :=
%   if 0 < r then ccw
%   else if r < 0 then cw
%   else collinear

Using the function $\sigma$ we can define the notion of \emph{general position} for collections (e.g., finite sets, lists, etc.) of points, simply postulating that $\sigma(p, q, r) \neq 0$ for every three distinct points $p, q, r$ in the collection.
Furthermore, we can start formalizing sets of points that are \emph{equivalent} with respect to their triple orientations, and consequently, properties of pointsets that are fully captured by their triple orientations~(\emph{orientation properties}).
% An illustration is presented in~\Cref{fig:sigma-equiv}.
% \input{fig-sigma-equiv.tex}
\begin{lstlisting}
structure σEquiv (S T : Set Point) where
  f : Point → Point
  bij : Set.BijOn f S T
  parity : Bool -- See Section 4 for details on this field
  σ_eq : ∀ (p ∈ S) (q ∈ S) (r ∈ S), σ p q r = parity ^^^ σ (f p) (f q) (f r)

def OrientationProperty (P : Set Point → Prop) :=
  ∀ {{S T}}, S ≃σ T → P S → P T -- `≃σ` is infix notation for `σEquiv`
\end{lstlisting}

To illustrate how these notions will be used, let us consider the property
\(
  \pi_k(S) \triangleq \text{\em ``pointset } S \text{ contains an empty convex } k\text{-gon''}
\), formalized as \lstinline|HasEmptyKGon|.

Based on \lstinline|ConvexEmptyIn.iff_triangles|, we know that $\pi_k(S)$ can be written
in terms of whether certain triangles are empty w.r.t $S$.
We can define triangle membership using $\sigma$,
and prove its equivalence to the geometric definition.
\begin{lstlisting}
/-- `Means that `a' is in the triangle `pqr', possibly on the boundary. -/
def PtInTriangle (a : Point) (p q r : Point) : Prop :=
  a ∈ convexHull ℝ {p, q, r}

/-- `Means that `a' is in the triangle `pqr' strictly, not on the boundary. -/
def σPtInTriangle (a p q r : Point) : Prop :=
  σ p q a = σ p q r ∧ σ p a r = σ p q r ∧  σ a q r = σ p q r

theorem σPtInTriangle_iff {a p q r : Point} (gp : InGenPos₄ a p q r) :
  σPtInTriangle a p q r ↔ PtInTriangle a p q r
\end{lstlisting}

% def PtInTriangle (a p q r : Point) : Prop := a ∈ convexHull ℝ {p, q, r}

% def σPtInTriangle (a p q r : Point) : Prop :=
%   σ p q r = σ p a r ∧ σ p a q = σ p r q ∧ σ q a r = σ q p r

% theorem σPtInTriangle_iff {a p q r : Point} (gp : PtFinsetInGenPos {a,p,q,r}) :
%   σPtInTriangle a p q r ↔ PtInTriangle a p q r -- not trivial.

% def HasEmptyTriangle (pts : Set Point) : Prop := ∃ p q r, [p, q, r].Nodup
% ∧ {p,q,r} ⊆ pts ∧ ∀ a ∈ pts, a ∉ ({p, q, r} : Set Point) → ¬PtInTriangle a p q r

% theorem OrientationProperty_HasEmptyTriangle : OrientationProperty HasEmptyTriangle

% Let us now discuss the previous steps. First,~\lstinline|(PtInTriangle a p q r)| presents a \emph{ground-truth}  definition for membership in a triangle, in terms of~\textsf{mathlib}'s \lstinline|ConvexHull| definition,  whereas~\lstinline|(σPtInTriangle a p q r)| is based on orientations.
Heule and Scheucher used the orientation-based definition~\cite{emptyHexagonNumber}, and as it is common in the area, its equivalence to the \emph{ground-truth} mathematical definition was left implicit.
This equivalence, formalized in~\lstinline|theorem σPtInTriangle_iff| is not trivial to prove:
the forward direction in particular requires reasoning about convex combinations and determinants.
Using the previous theorem, we can generalize to $k$-gons as follows.
\begin{lstlisting}
def σIsEmptyTriangleFor (a b c : Point) (S : Set Point) : Prop :=
  ∀ s ∈ S, ¬σPtInTriangle s a b c

def σHasEmptyKGon (n : Nat) (S : Set Point) : Prop :=
  ∃ s : Finset Point, s.card = n ∧ ↑s ⊆ S ∧ ∀ (a ∈ s) (b ∈ s) (c ∈ s), 
  a ≠ b → a ≠ c → b ≠ c → σIsEmptyTriangleFor a b c S

theorem σHasEmptyKGon_iff_HasEmptyKGon (gp : ListInGenPos pts) :
      σHasEmptyKGon n pts.toFinset ↔ HasEmptyKGon n pts.toFinset
\end{lstlisting}

Then, because \lstinline|σHasEmptyKGon| is ultimately defined in terms of $\sigma$, we can prove
\begin{lstlisting}
lemma OrientationProperty_σHasEmptyKGon : OrientationProperty (σHasEmptyKGon n)
\end{lstlisting}
Which in combination with \lstinline|theorem σHasEmptyKGon_iff_HasEmptyKGon|, provides
\begin{lstlisting}
theorem OrientationProperty_HasEmptyKGon : OrientationProperty (HasEmptyKGon n)
\end{lstlisting}

Let us discuss why the previous theorem is relevant, as it plays an important role in the formalization of Erd\H{o}s-Szekeres-type problems. This boils down to two reasons:
\begin{enumerate}
  \item If we prove that the function $\sigma$ is invariant under a certain transformation of its arguments (e.g., rotations, translations, etc.) then we can directly conclude that any orientation property is invariant under the same transformation. This is a powerful tool for applying manipulations to pointsets that preserve the properties of interest, which will be key for symmetry breaking (see~\Cref{sec:symmetry-breaking}).
    For a concrete example, consider a proof of an Erd\H{o}s-Szekeres-type result that starts by saying \emph{``we assume without loss of generality that points $p_1, \ldots, p_n$ all have positive $y$-coordinates''}.
    As translations are $\sigma$-equivalences, we can see that this assumption
    indeed does not impact the truth of any orientation property.
  \item As introduced at the beginning of this section, SAT encodings for Erd\H{o}s-Szekeres-type problems are based on capturing properties like convexity or emptiness in terms of triple orientations, thus reducing a continuous search space to a discrete one. Because we have proved that $\pi_k(S)$ is an orientation property, the values of $\sigma$ for all triples of points in $S$ contain enough information to determine whether $\pi_k(S)$ or not.  Therefore, we have proved that given $n$ points it is enough to analyze the values of $\sigma$ over these points, a discrete space with at most $2^{n^3}$ possibilities, instead of $\left(\mathbb{R}^2\right)^n$. This is the key idea that will allow us to transition from the finitely-verifiable statement \emph{``no set of triple orientations over $n$ points satisfies property $\pi_k$''} to \emph{``no set of $n$ points satisfies property $\pi_k$''}.
\end{enumerate}

\subsection{Properties of orientations}\label{sec:sigma-props}

We now prove,
assuming points are sorted left-to-right (which is justified in~\Cref{sec:symmetry-breaking}),
that certain \emph{$\sigma$-implication-properties} hold.
Consider four points $p, q, r, s$ with $p_x < q_x < r_x < s_x$.
If $p, q, r$ are oriented counterclockwise,
and $q, r, s$ are oriented counterclockwise as well,
then it follows that $p, r, s$ must be oriented counterclockwise
(see~\Cref{fig:orientation-implication}).
We prove a number of properties of this form:

\begin{lstlisting}
theorem σ_prop₁ (h : Sorted₄ p q r s) (gp : InGenPos₄ p q r s) :
    σ p q r = ccw → σ q r s = ccw → σ p r s = ccw

 [...]

theorem σ_prop₃ (h : Sorted₄ p q r s) (gp : InGenPos₄ p q r s) :
    σ p q r = cw → σ q r s = cw → σ p r s = cw
\end{lstlisting}

They will be used in justifying the addition of clauses~\labelcref{eq:signotopeClauses11,eq:signotopeClauses12};
clauses like these or the one below are easily added,
and are commonly used to reduce the search space in SAT encodings~\cite{emptyHexagonNumber,scheucherTwoDisjoint5holes2020,subercaseaux2023minimizing, szekeres_peters_2006}.

\begin{align}
  &\left(\neg \orvar_{a, b, c} \lor \neg \orvar_{a, c, d} \lor \orvar_{a, b, d}\right) \land \left(\orvar_{a, b, c} \lor \orvar_{a, c, d} \lor  \neg \orvar_{a, b, d}\right)
\end{align}

Our proofs of these properties are based on an equivalence between the orientation of a triple of points and the \emph{slopes} of the lines that connect them. Namely, if $p, q, r$  are sorted from left to right, then (i) $\sigma(p,q,r)=1 \iff \textsf{slope}(\overrightarrow{pq}) < \textsf{slope}(\overrightarrow{pr})$  and (ii) $\sigma(p,q,r)=1 \iff \textsf{slope}(\overrightarrow{pr}) < \textsf{slope}(\overrightarrow{qr})$. By first proving these \emph{slope-orientation} equivalences we can then easily prove \lstinline|σ_prop₁| and others, as illustrated in~\Cref{fig:orientation-implication}.

\begin{figure}
  \centering
  \begin{tikzpicture}
    \newcommand{\localspacing}{4.5}

    \foreach \i in {0, 1, 2} {

      \coordinate (p\i) at ( \i*\localspacing +0.5,0);
      \coordinate (q\i) at ( \i*\localspacing +2.5, 0.75);
      \coordinate (r\i) at ( \i*\localspacing +3.25, 1.5);
      \coordinate (s\i) at ( \i*\localspacing +4.0, 3.25);
    }
    \coordinate (0p) at (13,0);
    \coordinate (0q) at (13, 0.75);
    \coordinate (0r) at (13, 1.5);
    \coordinate (0s) at (13, 3.25);

    \pic [draw, latex-latex,
    angle radius=9mm, angle eccentricity=0.8, fill=blue!20!white,
    "$\text{\tiny{\(\theta_1\)}}$"] {angle = 0p--p0--q0};

    \pic [draw, latex-latex,
    angle radius=7mm, angle eccentricity=0.6, fill=orange!20!white,
    "$\text{\tiny{\(\theta_2\)}}$"] {angle = 0q--q0--r0};

    \pic [draw, latex-latex,
    angle radius=7mm, angle eccentricity=0.6, fill=orange!20!white,
    "$\text{\tiny{\(\theta_2\)}}$"] {angle = 0q--q1--r1};

    \pic [draw, latex-latex,
    angle radius=6mm, angle eccentricity=0.6, fill=yellow!20!white,
    "$\text{\tiny{\(\theta_3\)}}$"] {angle = 0r--r1--s1};

    \pic [draw, latex-latex,
    angle radius=9mm, angle eccentricity=0.8, fill=purple!20!white,
    "$\text{\tiny{\(\theta_4\)}}$"] {angle = 0p--p2--r2};

    \pic [draw, latex-latex,
    angle radius=6mm, angle eccentricity=0.6, fill=yellow!20!white,
    "$\text{\tiny{\(\theta_3\)}}$"] {angle = 0r--r2--s2};

    \foreach \i in {0, 1, 2} {

      \node[draw, circle, black, fill=blue, inner sep=0pt, minimum size=5pt] (p\i) at ( \i*\localspacing +0.5,0) {};
      \node[] at ( \i*\localspacing + 0.3, -0.25) {$p$};
      \node[draw, circle, black, fill=blue, inner sep=0pt, minimum size=5pt] (q\i) at ( \i*\localspacing +2.5, 0.75) {};
      \node[] at ( \i*\localspacing +2.6, 0.5) {$q$};
      \node[draw, circle, black, fill=blue, inner sep=0pt, minimum size=5pt] (r\i) at ( \i*\localspacing +3.25, 1.5) {};
      \node[] at ( \i*\localspacing +3.4, 1.3) {$r$};

      \node[draw, circle, black, fill=blue, inner sep=0pt, minimum size=5pt] (s\i) at ( \i*\localspacing +4.0, 3.25) {};
      \node[] at ( \i*\localspacing +3.75, 3) {$s$};
    }

    \draw[thick, black] (p0) -- (q0);
    \draw[thick, black, -latex] (q0) -- (r0);

    \draw[ thick, black] (q1) -- (r1);
    \draw[ thick, black, -latex] (r1) -- (s1);

    \draw[ thick, black] (p2) -- (r2);
    \draw[ thick, black, -latex] (r2) -- (s2);
    % \draw[  dashed, green!60!black] (0 + \localdx, 0 + \localdy) -- (2.5 + \localdx, 0.75 + \localdy);
    % \draw[ dashed, green!60!black, ->] (2.5 + \localdx, 0.75 + \localdy) -- (3.25 + \localdx, 1.5 + \localdy);

    % \draw[  dashed, green!60!black] (2.5  - \localdx, 0.75 - \localdy) -- (3.25 - \localdx, 1.5 - \localdy);
    % \draw[ dashed, green!60!black, ->] (3.25 - \localdx, 1.5 - \localdy) -- (4.25  - \localdx, 3.25 - \localdy);

    % \draw[ thick, dashed, green!60!black] (p) -- (r);
    % \draw[ thick, dashed, green!60!black, ->] (r) -- (s);

    \draw[dashed] (p0) -- (0p);
    \draw[dashed] (q0) -- (0q);
    \draw[dashed] (r0) -- (0r);
    \draw[dashed] (s0) -- (0s);

  \end{tikzpicture}
  \caption{Illustration for $\sigma(p,q,r) = 1 \; \land \; \sigma(q,r,s) = 1 \implies \sigma(p, r, s) = 1$. As we have assumptions $\theta_3 > \theta_2 > \theta_4$  by the forward direction of the \emph{slope-orientation equivalence}, we deduce $\theta_3 > \theta_4$, and then conclude $\sigma(p, r, s) = 1$ by the backward direction of the \emph{slope-orientation equivalence}.  }\label{fig:orientation-implication}
\end{figure}
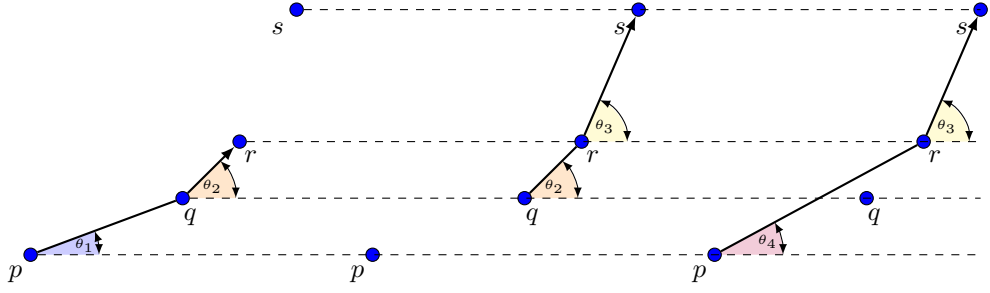

% \section{The Empty Triangle Theorem}\label{sec:empty-triangle}
% \input{empty_triangle.tex}

\section{Symmetry Breaking}\label{sec:symmetry-breaking}
\emph{Symmetry breaking} plays a key role in modern SAT-solving by substantially reducing the search space of assignments to a formula~\cite{biereHandbookSatisfiabilityVolume2009,Crawford}.
For example, if one proves that all satisfying assignments to a formula $\phi$ have either (i) $x_1 = 0, x_2 = 1$, or  (ii) $x_1 = 1, x_2 = 0$, and there is a bijection between satisfying assignments of form (i) and satisfying assignments of form (ii), then one can assume, \emph{without loss of generality}, that $x_1 = 0, x_2 = 1$, and thus add unit clauses $\ov{x_1}$ and $x_2$ to the formula $\phi$ while preserving its satisfiability.

In the context of the Empty Hexagon Number, the symmetry breaking done by Heule and Scheucher consists in assuming
that in order to search for a list of $30$ points in general position without a $6$-hole,
it suffices to can search only amongst lists of $30$ points in \emph{canonical} position.
These are defined as follows.
\begin{definition}[Canonical Position]
A list $L = (p_1,\ldots, p_{n})$ of points is said to be in canonical position if it satisfies all the following properties:
\begin{itemize}
    \item \textbf{($x$-order)} The points are sorted with respect to their $x$-coordinates, i.e., $x(p_i) < x(p_j)$ for all $1 \leq i < j \leq n$.
    \item \textbf{(General Position)} No three points are collinear, i.e., for all $1 \leq i < j < k \leq n$, we have $\sigma(p_i, p_j, p_k) \neq 0$.
    \item \textbf{(CCW-order)} All orientations $\sigma(p_1, p_i, p_j)$, with $1 < i < j \leq n$, are counterclockwise.
    \item \textbf{(Lex order)} The list of orientations \( \left(\sigma\left(p_{\lceil \frac{n}{2} \rceil -1}, p_{\lceil \frac{n}{2} \rceil},p_{\lceil \frac{n}{2} \rceil+1}\right), \ldots, \sigma\left(p_2, p_3, p_4\right) \right)\) is not lexicographically smaller than the list \(\left(\sigma\left(p_{\lfloor \frac{n}{2} \rfloor  + 1}, p_{\lfloor \frac{n}{2} \rfloor+2},p_{\lfloor \frac{n}{2} \rfloor+3}\right), \ldots, \sigma\left(p_{n-2}, p_{n-1}, p_{n}\right) \right).\)
    % Given the general position condition, all orientations are in $\{-1, 1\}$, and thus the lexicographic condition is equivalent to stating that there is an index $i$ such that $\forall j < i$, $\textsf{Left}[j] = \textsf{Right}[j]$, and $\textsf{Left}[i] = -1$ but $\textsf{Right}[i] = 1$.
\end{itemize}
\end{definition}

This symmetry breaking assumption not only reduces the search space of the SAT solver, but it is required for the correctness of the encoding,
as clauses~\labelcref{eq:insideClauses1,eq:insideClauses2,eq:holeDefClauses1,eq:signotopeClauses11,eq:signotopeClauses12} rely on points being sorted from left to right.
Before discussing the proof of correctness of symmetry breaking,
let us first focus on the last condition, expressed explicitly in clause~\labelcref{eq:revLexClauses} of the encoding.
The main idea behind this condition is that \emph{reflecting} a pointset, i.e., applying the map $(x, y) \mapsto (-x, y)$,
preserves the presence of $k$-holes, or convex $k$-gons.
This is the reason for incorporating the \emph{parity} flag in the definition of $\sigma$-equivalence.
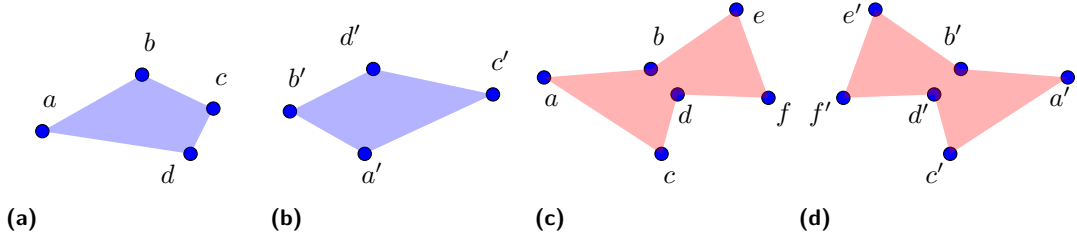
\begin{figure}
    \centering
    \begin{subfigure}{0.24\linewidth}
        \centering
        \begin{tikzpicture}
            \node[draw, circle, black, fill=blue, inner sep=0pt, minimum size=5pt, label={[xshift=0.1cm, yshift=0.1cm]$a$}] (a) at (0*0.75,0*0.75) {};
            \node[draw, circle, black, fill=blue, inner sep=0pt, minimum size=5pt, label={[xshift=0.1cm, yshift=0.1cm]$b$}] (b) at (1.75*0.75,1*0.75) {};
            \node[draw, circle, black, fill=blue, inner sep=0pt, minimum size=5pt, label={[xshift=0.1cm, yshift=0.1cm]$c$}] (c) at (3*0.75,0.4*0.75) {};
            \node[draw, circle, black, fill=blue, inner sep=0pt, minimum size=5pt, label={[xshift=-0.3cm, yshift=-0.6cm]$d$}] (d) at (2.6*0.75,-0.4*0.75) {};
            \coordinate (a) at (0*0.75,0*0.75);
            \coordinate (b) at (1.75*0.75,1*0.75);
            \coordinate (c) at (3*0.75,0.4*0.75);
            \coordinate (d) at (2.6*0.75,-0.4*0.75);
            \fill[blue, opacity=0.3] (a) -- (b) -- (c) -- (d) -- (a) -- cycle;
        \end{tikzpicture}
        \caption{}\label{fig:equiv-a}
    \end{subfigure}
    \begin{subfigure}{0.24\linewidth}
        \centering
        \begin{tikzpicture}[scale=0.75]
            \node[draw, circle, black, fill=blue, inner sep=0pt, minimum size=5pt, label={[xshift=0.1cm, yshift=-0.6cm]$a'$}] (a) at (0*0.75,0*0.75) {};
            \node[draw, circle, black, fill=blue, inner sep=0pt, minimum size=5pt, label={[xshift=0.1cm, yshift=0.1cm]$b'$}] (b) at (-1.75*0.75,1*0.75) {};
            \node[draw, circle, black, fill=blue, inner sep=0pt, minimum size=5pt, label={[xshift=0.1cm, yshift=0.1cm]$c'$}] (c) at (3*0.75,1.4*0.75) {};
            \node[draw, circle, black, fill=blue, inner sep=0pt, minimum size=5pt, label={[xshift=-0.3cm, yshift=0.1cm]$d'$}] (d) at (0.2*0.75,2*0.75) {};
            \coordinate (a) at (0*0.75,0*0.75);
            \coordinate (b) at (-1.75*0.75,1*0.75);
            \coordinate (c) at (3*0.75,1.4*0.75);
            \coordinate (d) at (0.2*0.75,2*0.75);*0.75
            \fill[blue, opacity=0.3] (b) -- (d) -- (c) -- (a) -- (b) -- cycle;
        \end{tikzpicture}
        \caption{}\label{fig:equiv-b}
    \end{subfigure}
%    \vspace{0.5cm}
    \begin{subfigure}{0.24\linewidth}
        \centering
        \begin{tikzpicture}[scale=.75]
            \node[draw, circle, black, fill=blue, inner sep=0pt, minimum size=5pt, label={[xshift=0.1cm, yshift=-0.6cm]$a$}] (a) at (1*2.5*0.75,6*0.4*0.75) {};
            \node[draw, circle, black, fill=blue, inner sep=0pt, minimum size=5pt, label={[xshift=0.1cm, yshift=0.1cm]$b$}] (b) at (2*2.5*0.75,6.5*0.4*0.75) {};
            \node[draw, circle, black, fill=blue, inner sep=0pt, minimum size=5pt, label={[xshift=0.1cm, yshift=-0.6cm]$c$}] (c) at (2.1*2.5*0.75,1.5*0.4*0.75) {};
            \node[draw, circle, black, fill=blue, inner sep=0pt, minimum size=5pt, label={[xshift=0.1cm, yshift=-0.6cm]$d$}] (d) at (2.25*2.5*0.75,5*0.4*0.75) {};
            \node[draw, circle, black, fill=blue, inner sep=0pt, minimum size=5pt, label={[xshift=0.3cm, yshift=-0.4cm]$e$}] (e) at (2.8*2.5*0.75,10*0.4*0.75) {};
            \node[draw, circle, black, fill=blue, inner sep=0pt, minimum size=5pt, label={[xshift=0.2cm, yshift=-0.6cm]$f$}] (f) at (3.1*2.5*0.75,4.8*0.4*0.75) {};
            \coordinate (a) at (1*2.5*0.75,6*0.4*0.75);
            \coordinate (b) at (2*2.5*0.75,6.5*0.4*0.75);
            \coordinate (c) at (2.1*2.5*0.75,1.5*0.4*0.75);
            \coordinate (d) at (2.25*2.5*0.75,5*0.4*0.75);
            \coordinate (e) at (2.8*2.5*0.75,10*0.4*0.75);
            \coordinate (f) at (3.1*2.5*0.75,4.8*0.4*0.75);
            \fill[red, opacity=0.3] (a) -- (b) -- (e) -- (f) -- (d) -- (c) -- (a) -- cycle;
        \end{tikzpicture}
        \caption{}\label{fig:equiv-c}
    \end{subfigure}
    \begin{subfigure}{0.24\linewidth}
        \centering
        \begin{tikzpicture}[scale=.75]
            \node[draw, circle, black, fill=blue, inner sep=0pt, minimum size=5pt, label={[xshift=-0.1cm, yshift=-0.6cm]$a'$}] (a) at (-1*2.5*0.75,6*0.4*0.75) {};
            \node[draw, circle, black, fill=blue, inner sep=0pt, minimum size=5pt, label={[xshift=-0.1cm, yshift=0.1cm]$b'$}] (b) at (-2*2.5*0.75,6.5*0.4*0.75) {};
            \node[draw, circle, black, fill=blue, inner sep=0pt, minimum size=5pt, label={[xshift=-0.2cm, yshift=-0.6cm]$c'$}] (c) at (-2.1*2.5*0.75,1.5*0.4*0.75) {};
            \node[draw, circle, black, fill=blue, inner sep=0pt, minimum size=5pt, label={[xshift=-0.2cm, yshift=-0.6cm]$d'$}] (d) at (-2.25*2.5*0.75,5*0.4*0.75) {};
            \node[draw, circle, black, fill=blue, inner sep=0pt, minimum size=5pt, label={[xshift=-0.3cm, yshift=-0.4cm]$e'$}] (e) at (-2.8*2.5*0.75,10*0.4*0.75) {};
            \node[draw, circle, black, fill=blue, inner sep=0pt, minimum size=5pt, label={[xshift=-0.3cm, yshift=-0.6cm]$f'$}] (f) at (-3.1*2.5*0.75,4.8*0.4*0.75) {};
            \coordinate (a) at (-1*2.5*0.75,6*0.4*0.75);
            \coordinate (b) at (-2*2.5*0.75,6.5*0.4*0.75);
            \coordinate (c) at (-2.1*2.5*0.75,1.5*0.4*0.75);
            \coordinate (d) at (-2.25*2.5*0.75,5*0.4*0.75);
            \coordinate (e) at (-2.8*2.5*0.75,10*0.4*0.75);
            \coordinate (f) at (-3.1*2.5*0.75,4.8*0.4*0.75);
            \fill[red, opacity=0.3] (a) -- (b) -- (e) -- (f) -- (d) -- (c) -- (a) -- cycle;
        \end{tikzpicture}
        \caption{}\label{fig:equiv-d}
    \end{subfigure}
  \caption{The pointsets depicted in \Cref{fig:equiv-a,fig:equiv-b} are $\sigma$-equivalent with \lstinline|parity := false|
  since the bijection $f$ defined by $(a,b,c,d) \mapsto (b', d', c', a')$ satisfies $\sigma(p_i, p_j, p_k) = \sigma(f(p_i), f(p_j), f(p_j))$ for every $\{p_i, p_j, p_k\} \subseteq \{a,b,c,d\}$.
On the other hand, no orientation-preserving bijection exists for~\Cref{fig:equiv-c,fig:equiv-d},
which are only $\sigma$-equivalent with \lstinline|parity := true|.}\label{fig:sigma-equiv}
  %We obtained~\Cref{fig:equiv-c,fig:equiv-d} computationally.
  \end{figure}
As illustrated in~\Cref{fig:sigma-equiv}, there are pointsets that are only $\sigma$-equivalent to their reflections
with \lstinline|parity := true|.
We are now ready to state the main symmetry breaking theorem and sketch its proof.

%     While the first 3 conditions are now arguably standard in computational results regarding Erd\H{o}s-Szekeres type problems~\cite{scheucherTwoDisjoint5holes2020}, the last condition is a novelty introduced by Heule and Scheucher.
%     Interestingly, in the process of verifying the correctness of this symmetry-breaking assumption, we found a small error in the proof presented in~\cite{scheucherTwoDisjoint5holes2020} for the first $3$ conditions.
% The concrete theorem we prove is the following:
\input{fig-symmetry-breaking.tex}
\begin{lstlisting}
theorem symmetry_breaking : ListInGenPos l →
  ∃ w : CanonicalPoints, Nonempty (l.toFinset ≃σ w.points.toFinset)
\end{lstlisting}

\begin{proof}[Proof Sketch]
The proof proceeds in 6 steps, illustrated in~\Cref{fig:symmetry-breaking}. In each of the steps, we will construct a new list of points that is $\sigma$-equivalent to the previous one, and the last one will be in canonical position.\footnote{Even though we defined $\sigma$-equivalence for sets of points, our formalization goes back and forth between sets and lists. Given that symmetry breaking distinguishes between the order of the points e.g., $x$-order, this proof proceeds over lists. All permutations of a list are immediately $\sigma$-equivalent.}
The main justification for each step is that, given that the function $\sigma$ is defined as a sign of the determinant, applying transformations that preserve (or, when \lstinline|parity := true|, uniformly reverse) the sign of the determinant will preserve (or uniformly reverse) the values of $\sigma$. In particular, given the identity $\det(AB) = \det(A)\det(B)$, if we apply a transformation to the points that corresponds to multiplying by a matrix $B$ such that $\det(B) > 0$, then $\sign(\det(A)) = \sign(\det(AB))$, and thus orientations will be preserved.
In step 1, we transform the list of points so that no two points share the same $x$-coordinate. This can be done by applying a rotation to the list of points, which corresponds to multiplying by a rotation matrix.
Rotations always have determinant $1$. 
In step 2, we translate all points by a constant vector $t$, by multiplying by a translation matrix, so that the left most point gets position $(0, 0)$, and naturally every other point will have a positive $x$-coordinate.
Let $L_2$ be the list of points after this transformation, excluding $(0,0)$ which we will denote by $p_1$.
Then, in step 3, we  apply the projective transformation $f: (x, y) \mapsto (y/x, 1/x)$ to every point in $L_2$, showing that this preserves orientations within $L_2$.
To see that this mapping is a $\sigma$-equivalence consider that 
\[
\begin{multlined}
 \sign \det \begin{pmatrix} p_x & q_x & r_x \\ p_y & q_y & r_y \\ 1 & 1 & 1 \end{pmatrix} =  \sign \det \left( \begin{pmatrix} 0 & 0 & 1 \\ 1 & 0 & 0\\ 0 & 1 & 0 \end{pmatrix}  \begin{pmatrix} \nicefrac{p_y}{p_x} & \nicefrac{q_y}{q_x} & \nicefrac{r_y}{r_x} \\ \nicefrac{1}{p_x} & \nicefrac{1}{q_x} & \nicefrac{1}{r_x} \\ 1 & 1 & 1 \end{pmatrix}  \begin{pmatrix} p_x & 0 & 0 \\ 0 & q_x & 0\\ 0 & 0 & r_x \end{pmatrix} \right)\\
                        = \sign \left(1 \cdot \det  \begin{pmatrix} \nicefrac{p_y}{p_x} & \nicefrac{q_y}{q_x} & \nicefrac{r_y}{r_x} \\ \nicefrac{1}{p_x} & \nicefrac{1}{q_x} & \nicefrac{1}{r_x} \\ 1 & 1 & 1 \end{pmatrix} \cdot  p_x q_x r_x  \right) = \sign \det \begin{pmatrix} p_y/p_x & q_y/q_x & r_y/r_x \\ 1/p_x & 1/q_x & 1/r_x \\ 1 & 1 & 1 \end{pmatrix}.
                        %  \tag{As $p_x q_x r_x > 0$ by step 2}
\end{multlined}
\]
To preserve orientations with respect to the leftmost point $(0, 0)$, we set $f( (0, 0)) = (0, \infty)$, a special point that is treated separately as follows.
As the function $\sigma$ takes points in $\mathbb{R}^2$ as arguments,
we need to define an extension
\(
  \sigma_{(0, \infty)}(q, r) = \begin{cases}
    1 & \text{if } q_x < r_x \\
    -1 & \text{otherwise}.  
  \end{cases},
\)
We then show that $\sigma((0, 0), q, r) = \sigma_{(0, \infty)}(f(q), f(r))$ for all points $q, r \in L_2$. 
In step 4, we sort the list $L_2$ by $x$-coordinate in increasing order, thus obtaining a list $L_3$.
This can be done while preserving $\sigma$-equivalence because sorting corresponds to a permutation, and all permutations of a list are $\sigma$-equivalent by definition.
In step 5, we check whether the \textsf{Lex order} condition above is satisfied in $L_3$, and if it is not, we reflect the pointset, which as explained above, preserves $\sigma$-equivalence by leveraging the parity option in the definition. Note that in such a case we need to relabel the points from left to right again.
In step 6, we bring point $(0, \infty)$ back into the range by first finding a constant $c$ such that all points in $L_3$ are to the right of the line $y=c$, and then finding a large enough value $M$ such that $(c, M)$ has the same orientation with respect to the other points as $(0, \infty)$ did, meaning that 
\(\sigma((c, M), q, r) = \sigma_{(0, \infty)}(q, r)\) for every $q, r \in L_3$.
Finally, we note that the list of points obtained in step 6 satisfies the \text{CCW-order} property by the following reasoning:
if $1 < i < j \leq n$ are indices, then 
\begin{align*}
  \sigma(p_1, p_i, p_j) = 1 &\iff \sigma((c, M), p_i, p_j) = 1\\
                            &\iff \sigma_{(0, \infty)}(p_i, p_j) = 1\tag{By step 6}\\
                            &\iff (p_i)_x < (p_j)_x \tag{By definition of $\sigma_{(0, \infty)}$}\\
                            &\iff \textsf{true} \tag{By step 4, since points are sorted and $i < j$}.
\end{align*}
This concludes the proof.
\end{proof}

\section{The Encoding and Its Correctness}\label{sec:encoding}
Having established the reduction to orientations,
and the symmetry-breaking assumption of canonicity,
we now turn to the construction of a CNF formula $\phi_n$
whose unsatisfiability would imply
that every set of $n$ points
contains a $6$-hole.\footnote{
  Satisfiability of $\phi_n$ would \emph{not} necessarily imply
  the existence of a point set without a $6$-hole, due to the \emph{realizability problem} (see e.g.,~\cite{subercaseaux2023minimizing}).
  % First, the mapping from sets of points to assignments of triple orientations is not surjective.
  % Second, even if it were, $\phi_n$ is tailored to looking for an unsatisfiability proof,
  % using implication rather than bi-implication in some of the variable-defining clauses.
}
The formula is detailed in~\Cref{fig:full-encoding}.

\begin{figure}
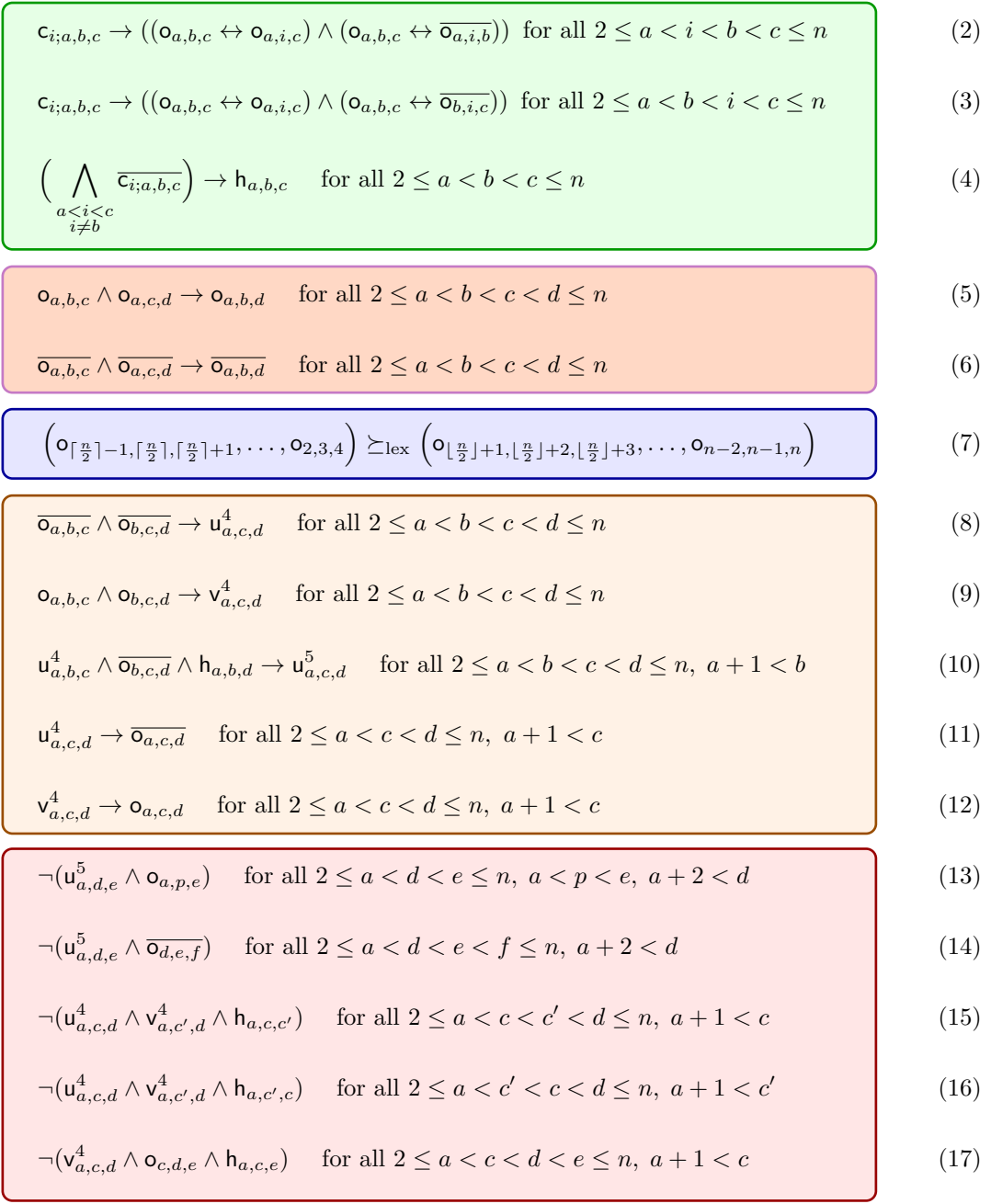

  \label{fig:full-encoding}
% \begin{framed}
  \begin{spreadlines}{16pt}
\begin{gather}
\hfsetfillcolor{green!10}
\hfsetbordercolor{green!60!black}
\tikzmarkin{b}(12.0,-0.9)(-0.5,0.5)
  \cvar_{i; a,b, c} \rightarrow \left(\left(\orvar_{a,b,c} \leftrightarrow \orvar_{a, i, c}  \right) \land \left(\orvar_{a,b,c} \leftrightarrow \ov{\orvar_{a, i, b}}  \right)\right) \text{ for all } 2 \leq a < i < b < c \leq n\label{eq:insideClauses1}\\
  \cvar_{i; a,b, c} \rightarrow \left(\left(\orvar_{a,b,c} \leftrightarrow \orvar_{a, i, c}  \right) \land \left(\orvar_{a,b,c} \leftrightarrow \ov{\orvar_{b, i, c}}  \right)\right) \text{ for all } 2 \leq a < b < i < c \leq n\label{eq:insideClauses2}\\
  \tikzmarkend{b}\Big(\bigwedge_{\substack{a < i < c\\ i \neq b}} \ov{\cvar_{i; a,b,c}}\Big) \rightarrow \hvar_{a, b, c} \quad \text{ for all } 2 \leq a < b < c \leq n\label{eq:holeDefClauses1}\\
\tikzmarkin{a}(12.0,-0.3)(-0.5,0.5)
  \orvar_{a, b, c} \land \orvar_{a, c, d} \rightarrow \orvar_{a, b, d} \quad \text{ for all } 2 \leq a < b < c < d \leq n\label{eq:signotopeClauses11}\\
  \tikzmarkend{a}\ov{\orvar_{a, b, c}} \land \ov{\orvar_{a, c, d}} \rightarrow \ov{\orvar_{a, b, d}} \quad \text{ for all } 2 \leq a < b < c < d \leq n \label{eq:signotopeClauses12}\\
\hfsetfillcolor{blue!10}
\hfsetbordercolor{blue!60!black}
\tikzmarkin{c}(12.0,-0.4)(-0.5,0.6)
  % \orvar_{1, b, c} \quad \text{ for all } 2 \leq b < c \leq n \label{eq:revLexClauses}\\
  \tikzmarkend{c}\left(\orvar_{\lceil \frac{n}{2} \rceil -1, \lceil \frac{n}{2} \rceil,\lceil \frac{n}{2} \rceil+1}, \ldots, \orvar_{2,3,4} \right) \succeq_{\text{lex}} \left(\orvar_{\lfloor \frac{n}{2}\rfloor +1,  \lfloor \frac{n}{2}\rfloor +2, \lfloor \frac{n}{2}\rfloor +3}, \ldots, \orvar_{n-2, n-1, n} \right)\label{eq:revLexClauses}\\
\hfsetfillcolor{orange!10}
\hfsetbordercolor{orange!60!black}
\tikzmarkin{d}(12.0,-0.3)(-0.5,0.5)
  \ov{\orvar_{a,b,c}} \land \ov{\orvar_{b,c,d}} \rightarrow \uvar_{a, c, d} \quad \text{ for all } 2 \leq a < b < c < d \leq n\label{eq:capDef}\\
  \orvar_{a, b, c} \land \orvar_{b, c, d} \rightarrow \vvar_{a, c, d} \quad \text{ for all } 2 \leq a < b < c < d \leq n \label{eq:cupDef}\\
  \uvar_{a,b,c} \land \ov{\orvar_{b,c,d}} \land \hvar_{a,b,d} \rightarrow \ufvar_{a, c, d} \quad \text{ for all } 2 \leq a < b < c < d \leq n,\; a+1<b\label{eq:capFDef}\\
  \uvar_{a, c, d} \rightarrow \ov{\orvar_{a,c,d}} \quad \text{ for all } 2 \leq a < c < d \leq n,\ a+1<c\label{eq:capDef2}\\
  \tikzmarkend{d}\vvar_{a, c, d} \rightarrow \orvar_{a,c,d} \quad \text{ for all } 2 \leq a < c < d \leq n,\; a+1<c\label{eq:cupDef2}\\
\hfsetfillcolor{red!10}
\hfsetbordercolor{red!60!black}
\tikzmarkin{e}(12.0,-0.5)(-0.5,0.5)
  \neg(\ufvar_{a,d,e} \land \orvar_{a, p, e}) \quad \text { for all } 2 \leq a < d < e \leq n, \; a < p < e, \; a+2 < d\label{eq:no6Hole1Below}\\
  \neg(\ufvar_{a,d,e} \land \ov{\orvar_{d, e, f}}) \quad \text { for all } 2 \leq a < d < e < f\leq n, \; a+2 < d\label{eq:no6Hole4Above}\\
  \neg(\uvar_{a,c,d} \land \vvar_{a, c', d} \land \hvar_{a,c,c'}) \quad \text{ for all } 2 \leq a < c < c' < d \leq n, \; a+1 < c\label{eq:no6Hole2Below1}\\
  \neg(\uvar_{a,c,d} \land \vvar_{a, c', d} \land \hvar_{a,c',c}) \quad \text{ for all } 2 \leq a < c' < c < d \leq n, \; a+1 < c'\label{eq:no6Hole2Below2}\\
  \tikzmarkend{e}\neg(\vvar_{a,c,d} \land \orvar_{c, d, e} \land \hvar_{a,c,e}) \quad \text{ for all } 2 \leq a < c < d < e \leq n, \; a+1 < c\label{eq:no6Hole3Below}
  \end{gather}
\end{spreadlines}
% \end{framed}
\caption{Encoding based on that of Heule and Scheucher for the Empty Hexagon Number~\cite{emptyHexagonNumber}. Each line determines a set of clauses. Unsatisfiability of the formula below for $n=30$ implies $h(6) \leq 30$, as detailed throughout the paper.}
\end{figure}

\subparagraph*{Variables.}
Let $S = (p_1, \ldots, p_n)$ be the list of points in canonical position.
We explain the variables of $\phi_n$
by specifying their values in the propositional assignment $\tau_S$
that is our intended model of $\phi_n$
corresponding to $S$. We then have:
\begin{itemize}
  \item
    For every $2 \leq a < b < c \leq n$, $\orvar_{a,b,c}$ is true
    iff $\sigma(p_a,p_b,p_c) = +1$.\footnote{
    Since the point set is in general position,
    we have $\neg \orvar_{a,b,c} \iff \sigma(p_a, p_b, p_c) = -1$.}

    The first optimization observes that orientations are antisymmetric:
    if $(p,q,r)$ is counterclockwise then $(q,p,r)$ is clockwise, etc.
    Thus one only needs $\orvar_{a,b,c}$ for ordered triples $(a,b,c)$,
    reducing the number of orientation variables by a factor of $3! = 6$
    relative to using all triples. The second optimization uses the \textbf{CCW-order} property of canonical positions:
    since all $\orvar_{1,a,b}$ are true, we may as well omit them from the encoding.

  \item
    Next, for every $a < b < c$ with $a < i < b$ or $b < i < c$,
    the variable $\cvar_{i;a,b,c}$ is true
    iff \lstinline|σPtInTriangle S[i] S[a] S[b] S[c]| holds.
    By \lstinline|σPtInTriangle_iff|, this is true exactly
    iff $p_i$ is inside the triangle $p_ap_bp_c$.
    The reason for assuming $(a,b,c)$ to be ordered is again symmetry:
    $p_ap_bp_c$ is the same triangle as $p_ap_cp_b$, etc.
    Furthermore thanks to the \textbf{$x$-order} property of canonical positions,
    if $p_i$ is in the triangle
    then $x(p_a) < x(p_i) < x(p_c)$.
    This implies that $a < i < c$,
    leaving one case distinction permuting $(i,b)$.

  \item
    For every $a < b < c$,
    $\hvar_{a,b,c}$ is true
    iff \lstinline|σIsEmptyTriangleFor S[a] S[b] S[c] S| holds.
    %  Analogously to the previous items,
    By a geometro-combinatorial connection analogous to ones above,
    this is true iff $p_ap_bp_c$ is a $3$-hole.

  \item
    Finally, one defines \emph{$4$-cap}, \emph{$5$-cap}, and \emph{$4$-cup} variables.
    For $a+1 < c < d$, $\uvar_{a,c,d}$ is true
    iff there is $b$ with $a < b < c$ with $\sigma(p_a,p_b,p_c) = \sigma(p_b,p_c,p_d) = -1$.
    $\vvar_{a,c,d}$ is analogous, except in that the two orientations are required to be counterclockwise.
    These are the $4$-caps and $4$-cups, respectively.
    The $5$-cap variables $\ufvar_{a,d,e}$
    are defined for $a+2 < d < e$.
    We set $\ufvar_{a,d,e}$ to true
    iff there exists $c$ with $a+1<c<d$
    such that $\uvar_{a,c,d}$, $\orvar_{c,d,e}$, and $\hvar_{a,c,e}$ are all true.
    Intuitively, $4$-caps and $4$-cups are clockwise and counterclockwise arcs of length $4$,
    respectively,
    whereas $5$-caps are clockwise arcs of length $5$ containing a $3$-hole.
    All three are depicted in~\Cref{fig:cup-cap-vars}. The usage of these variables is crucial to an efficient encoding:
    we will show below that a hexagon can be covered by only $4$ triangles,
    so one need not consider all ${6\choose 3}$ triangles contained within it.
\end{itemize}

\input{fig-cup-cap.tex}

\subparagraph*{Satisfaction.}
We now have to justify that the clauses of $\phi_n$
are satisfied by the intended interpretation $\tau_S$
for a $6$-hole-free point set $S$.
The variable-defining clauses~\labelcref{eq:insideClauses1,eq:insideClauses2,eq:holeDefClauses1,eq:capDef,eq:cupDef,eq:capFDef,eq:capDef2,eq:cupDef2}
follow essentially by definition combined with boolean reasoning.
The orientation properties~\labelcref{eq:signotopeClauses11,eq:signotopeClauses12}
have been established in the family of theorems \lstinline|σ_propᵢ|.
The lexicographic ordering clauses~\labelcref{eq:revLexClauses}
follow from the \textbf{Lex order} property of canonical positions.
Thus we are left with clauses~\labelcref{eq:no6Hole1Below,eq:no6Hole4Above,eq:no6Hole2Below1,eq:no6Hole2Below2,eq:no6Hole3Below}
which forbid the presence of certain $6$-holes.\footnote{
They are intended to forbid \emph{all} $6$-holes,
but proving completeness is not necessary for an unsatisfiability-based result.}
We illustrate why clause~\labelcref{eq:no6Hole1Below} is true.
The contrapositive is easier to state:
if $\tau_S$ satisfies $\ufvar_{a,d,e}\wedge\orvar_{a,p,e}$,
then $S$ contains a $6$-hole.
The intuitive argument is depicted in~\Cref{fig:clause-13-forbid}.
The clause directly implies the existence of a convex hexagon $apedcb$
such that $ace$ is a $3$-hole.
It turns out that this is enough to ensure
the existence of a $6$-hole by ``flattening'' the triangles $ape$, $edc$, and $cba$,
if necessary,
to obtain empty triangles $ap'e$, $ed'c$, and $cb'a$,
which can be assembled into a $6$-hole $ap'ed'cb'$.

Justifying this formally turned out to be complex,
requring a fair bit of reasoning about point \lstinline|Arc|s
and \lstinline|σCCWPoints|: lists of points winding around a convex polygon.
Luckily, the main argument can be summarized in terms of two facts:
(a) any triangle $abc$ contains an empty triangle $ab'c$; and
(b) empty shapes sharing a common line segment can be glued together.
Formally, (a) can be stated as

\begin{lstlisting}
theorem σIsEmptyTriangleFor_exists (gp : ListInGenPos S)
  (abc : [a, b, c] ⊆ S) : ∃ b' ∈ S, σ a b' c = σ a b c
    ∧ (b' = b ∨ σPtInTriangle b' a b c) ∧ σIsEmptyTriangleFor a b' c S.toFinset
\end{lstlisting}
\begin{proof}
  Given points $p,q$, say that $p \leq q$ iff $p$ is in the triangle $aqc$.
  This is a preorder.
  Now, the set $S' = \{x \in S \mid \sigma(a,x,c) = \sigma(a,b,c) \wedge x \leq b\}$
  is finite and so has a weakly minimal element $b'$,
  in the sense that no $x \in S'$ has $x < b'$.
  Emptiness of $ab'c$ follows by minimality.
\end{proof}

Moving on, (b) follows from a \emph{triangulation lemma}:
given any convex point set $S$
and a line $\overleftrightarrow{ab}$ between two vertices of $S$,
the convex hull of $S$ is contained in the convex hulls
of points on either side of $\overleftrightarrow{ab}$.
That is:
\begin{lstlisting}
theorem split_convexHull (cvx : ConvexPoints S) :
  ∀ {a b}, a ∈ S → b ∈ S →
    convexHull ℝ S ⊆ convexHull ℝ {x ∈ S | σ a b x ≠ ccw}
                    ∪ convexHull ℝ {x ∈ S | σ a b x ≠ cw}
\end{lstlisting}
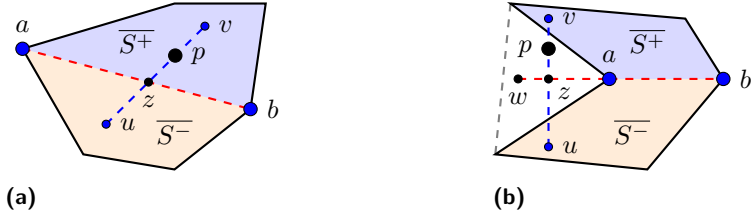
\begin{figure}
    \centering
    \begin{subfigure}{0.45\textwidth}
        \begin{tikzpicture}
        \coordinate (a) at (0, 0.4);
        \coordinate (b) at (3, -0.4);
        \coordinate (v1) at (3.2, 1);
        \coordinate (v2) at (2, 1);
        \coordinate (u1) at (0.8, -1);
        \coordinate (u2) at (2, -1.2);
        \coordinate (u) at (1.1, -0.6);
        \coordinate (v) at (2.4, 0.7);
        \coordinate (p) at (2.01,0.31);
        \coordinate (z) at (1.65789, -0.0421053);

        \fill[blue, opacity=0.15] (b) -- (v1) -- (v2) -- (a) -- cycle;
        \fill[orange, opacity=0.15] (a) -- (u1) -- (u2) -- (b) -- cycle;

        \draw[thick] (a) -- (u1) -- (u2) -- (b) -- (v1) -- (v2) -- (a);
        \draw[dashed, thick, red] (a) -- (b);
        \draw[dashed, thick, blue] (u) -- (v);

        \node[draw, circle, black, fill=blue, inner sep=0pt, minimum size=5pt, label=above:$a$] (pA) at (a) {};
        \node[draw, circle, black, fill=blue, inner sep=0pt, minimum size=5pt, label=right:$b$] (pB) at (b) {};
        \node[draw, circle, black, fill=blue, inner sep=0pt, minimum size=3pt, label=right:$u$] (pU) at (u) {};
        \node[draw, circle, black, fill=blue, inner sep=0pt, minimum size=3pt, label=right:$v$] (pV) at (v) {};

        \node[draw, circle, black, fill=black, inner sep=0pt, minimum size=5pt, label=right:$p$] (pP) at (p) {};
        \node[draw, circle, black, fill=black, inner sep=0pt, minimum size=3pt, label=below:$z$] (pZ) at (z) {};

        \node[] at (1.5, 0.5) {$\overline{S^+}$};
        \node[] at (2, -0.7) {$\overline{S^-}$};

        \end{tikzpicture}
        \caption{}\label{fig:triangulation2-a}
    \end{subfigure}
    \begin{subfigure}{0.45\textwidth}
        \begin{tikzpicture}
            \coordinate (a) at (1.5, 0);
            \coordinate (b) at (3, 0);
            \coordinate (v1) at (2.5, 0.8);
            \coordinate (v2) at (0.2, 1);
            \coordinate (u1) at (0, -1);
            \coordinate (u2) at (2, -1.2);
            \coordinate (u) at (0.7, -0.9);
            \coordinate (v) at (0.7, 0.8);
            \coordinate (p) at (0.7, 0.4);
            \coordinate (z) at (0.7, 0);
            \coordinate (w) at (0.3, 0);

            \fill[blue, opacity=0.15] (b) -- (v1) -- (v2) -- (a) -- cycle;
            \fill[orange, opacity=0.15] (a) -- (u1) -- (u2) -- (b) -- cycle;

            \draw[thick] (a) -- (u1) -- (u2) -- (b) -- (v1) -- (v2) -- (a);
            \draw[dashed, thick, red] (a) -- (b);
            \draw[dashed, thick, blue] (u) -- (v);
            \draw[dashed, thick, red] (w) -- (a);
            \draw[dashed, thick, opacity=0.5] (v2) -- (u1);

            \node[draw, circle, black, fill=blue, inner sep=0pt, minimum size=5pt, label=above:$a$] (pA) at (a) {};
            \node[draw, circle, black, fill=blue, inner sep=0pt, minimum size=5pt, label=right:$b$] (pB) at (b) {};
            \node[draw, circle, black, fill=blue, inner sep=0pt, minimum size=3pt, label=right:$u$] (pU) at (u) {};
            \node[draw, circle, black, fill=blue, inner sep=0pt, minimum size=3pt, label=right:$v$] (pV) at (v) {};

            \node[draw, circle, black, fill=black, inner sep=0pt, minimum size=5pt, label=left:$p$] (pP) at (p) {};
            \node[draw, circle, black, fill=black, inner sep=0pt, minimum size=3pt, label={[shift={(-0.05,0.05)}]-45:$z$}] (pZ) at (z) {};
            \node[draw, circle, black, fill=black, inner sep=0pt, minimum size=3pt, label=below:$w$] (pW) at (w) {};

            \node[] at (2, 0.5) {$\overline{S^+}$};
            \node[] at (1.8, -0.7) {$\overline{S^-}$};

        \end{tikzpicture}
        \caption{}\label{fig:triangulation2-b}
    \end{subfigure}
    \caption{Illustration of the proof for \lstinline|split_convexHull|. (a) Given point $p$, we obtain points $u$ and $v$ inside the two halves and $z$ as the point of intersection with the line $\overline{ab}$. (b) In this (contradictory) situation, the point $z$ has ended up outside the segment $\overline{ab}$, because $S$ is not actually convex. In this case we construct $w$ such that $z$ is on the $\overline{wa}$ segment, and observe that $w,z,a,b$ are collinear.}\label{fig:triangulation2}
\end{figure}
\begin{proof}
    Let $S^+=\{x\in S\mid \sigma(a,b,x)\ge 0\}$ and $S^-=\{x\in S\mid \sigma(a,b,x)\le 0\}$ be the two sets in the theorem, and let $p\in \overline{S}$, where $\overline{S}$ denotes the convex hull of $S$. Assume WLOG that $\sigma(a,b,p)\ge 0$. (We would like to show that $p\in \overline{S^+}$.) Now $p$ is a convex combination of elements of $S^+$ and elements of $S^-$, so there exist points $u\in \overline{S^-}$ and $v\in \overline{S^+}$ such that $p$ lies on the $\overline{uv}$ line.
    Because $\{x\mid \det(a,b,x)\le 0\}\supseteq S^-$ is convex, it follows that $\det(a,b,u)\le 0$, and likewise $\det(a,b,v)\ge 0$, so they lie on opposite sides of the $\overleftrightarrow{ab}$ line and hence $\overline{uv}$ intersects $\overleftrightarrow{ab}$ at a point $z$. The key point is that $z$ must in fact be on the line segment $\overline{ab}$; assuming that this was the case, we could obtain $z$ as a convex combination of $a$ and $b$, and $p$ as a convex combination of $v$ and $z$, and since $v$ is in $\overline{S^+}$ and $a,b\in S^+\subseteq\overline{S^+}$ we can conclude $p\in \overline{S^+}$.
    To show that $z\in \overline{ab}$, suppose not, so that $a$ lies between $z$ and $b$ (see \Cref{fig:triangulation2-b}). (The case where $z$ is on the $b$ side is similar.) We can decompose $z$ as a convex combination of some $w\in \overline{S\setminus\{a\}}$ and $a$, which means that $w,z,a,b$ are collinear and appear in this order on the line. Therefore $a$ is a convex combination of $w$ and $b$, which means that $a\in \overline{S\setminus\{a\}}$ which violates convexity of $S$.
\end{proof}

\noindent
By contraposition,
the triangulation lemma directly implies that
if $\{x \in S \mid \sigma(a,b,x) \neq +1\}$ and $\{x \in S \mid \sigma(a,b,x) \neq -1\}$
are both empty shapes in $P$,
then $S$ is an empty shape in $P$.

\subparagraph*{Running the CNF.}
Having now shown that our main result follows if $\phi_{30}$ is unsatisfiable,
we run a distributed computation to check its unsatisfiability.
We solve the SAT formula~$\phi_{30}$ produced by Lean using the same setup as
Heule and Scheucher~\cite{emptyHexagonNumber}, although using different hardware:
the Bridges 2 cluster of the Pittsburgh Supercomputing Center~\cite{cluster}.
Following Heule and Scheucher,
we partition the problem into 312\,418 subproblems.
Each of these subproblems was
solved using {\tt CaDiCaL} version 1.9.5.
The solver produced an LRAT proof for each execution,
which was validated using the {\tt cake\_lpr} verified checker on-the-fly
in order to avoid writing/storing/reading large files.
The total runtime was 25\,876.5 CPU hours, or roughly 3 CPU years.
The difference in runtime relative to Heule and Scheucher's original run
is purely due to the difference in hardware.
Additionally,
we validated that the subproblems cover the entire search space as Heule and Scheucher did~\cite[Section 7.3]{emptyHexagonNumber}.
This was done by verifying the unsatisfiability
of another formula that took 20 seconds to solve.

% file-local attic:

% By using the triangulation lemma repeatedly,
% we can slice up a convex $k$-gon into smaller pieces in any way we like until we get to triangles,
% thereby showing that convex polygons can be triangulated.

% \section{SAT Pipeline}\label{sec:leansat}
% \input{leansat.tex}

\section{Related Work}\label{sec:related-work}
Our formalization is closely related to a prior development
in which Marić put proofs of $g(6) \leq 17$ on a more solid foundation~\cite{19maric_fast_formal_proof_erdos_szekeres_conjecture_convex_polygons_most_six_points}.
The inequality,
originally obtained by Szekeres and Peters \cite{06szekeres_computer_solution_17_point_erdos_szekeres_problem}
using a specialized, unverified search algorithm,
was confirmed by Marić using a formally-verified SAT encoding.
Marić introduced an optimized encoding based on nested convex hull structures,
which, when combined with performance advances in modern SAT solvers,
significantly improved the search time over the unverified computation.

Our work focuses on the closely-related problem
of determining $k$-hole numbers $h(k)$.
Rather than devise a new SAT encoding,
we use essentially the same encoding presented by Heule and Scheucher~\cite{emptyHexagonNumber}.
Interestingly,
a (verified) proof of $g(6) \leq 17$ can be obtained
as a corollary of our development.
We can assert the hole variables $\hvar_{a,b,c}$ as true
while leaving the remainder of the encoding in~\Cref{fig:full-encoding} unchanged,
which trivializes constraints about emptiness
so that only the convexity constraints remain.\footnote{
This modification was performed by an author
who did not understand this part of the proof,
nevertheless having full confidence in its correctness
thanks to the Lean kernel having checked every assertion.}
The resulting CNF formula
asserts the existence of a set of $n$ points
with no convex $6$-gon.
We checked this formula to be unsatisfiable for $n = 17$,
giving the same result as Marić:
\begin{lstlisting}
theorem gon_6_theorem (pts : List Point) (gp : ListInGenPos pts)
    (h : pts.length ≥ 17) : HasConvexKGon 6 pts.toFinset
\end{lstlisting}

Since both formalizations can be executed,
we performed a direct comparison against Marić's encoding.
On a personal laptop,
we found that it takes negligible time (below 1s)
for our verified Lean encoder to output the full CNF.
In contrast,
Marić's encoder, extracted from Isabelle/HOL code,
took 437s to output a CNF
(this was compiled on Isabelle/HOL 2016,
the latest version that accepts the codebase without broader changes).
To circumvent the encoder slowness,
Marić wrote a C++ encoder
whose code was manually compared against the Isabelle/HOL specification.
We do not need to resort to an unverified implementation.

As for the encodings,
ours took 28s to solve,
while the Marić encoding took 787s (both using \textsf{cadical}).
This difference is likely accounted for in the relative size of the encodings,
in particular their symmetry breaking strategies.
For $k=6$ and $n$ points,
the encoding of Heule and Scheucher uses $O(n^4)$ clauses,
whereas the one of Marić uses $O(n^6)$ clauses.
They are based on different ideas:
the former as detailed in~\Cref{sec:symmetry-breaking},
whereas the latter on nested convex hulls.
The different approaches have been discussed by Scheucher \cite{scheucherTwoDisjoint5holes2020}.
This progress in solve times
represents an encouraging state of affairs;
we are optimistic that if continued,
it could lead to an eventual resolution of $g(7)$.

Further differences include what exactly was formally proven.
As with most work in this area,
we use the combinatorial abstraction of triple orientations.
We and Marić alike show that point sets in $\mathbb R^2$
satisfy orientation properties (\Cref{sec:triple-orientations}).
However, our work goes further in building the connection
between geometry and combinatorics:
our definitions of convexity and emptiness (\Cref{sec:geometry}),
and consequently the theorem statements,
are geometric ones based on convex hulls
as defined in Lean's \texttt{mathlib}~\cite{The_mathlib_Community_2020}.
In contrast, Marić axiomatizes these properties in terms of $\sigma$.
A skeptical reviewer must manually verify that these combinatorial definitions
correspond to the desired geometric concept.

A final point of difference concerns the verification of SAT proofs.
Marić fully reconstructs some of the SAT proofs on which his results depend,
though not the main one for $g(6)$,
in an NbE-based proof checker for Isabelle/HOL.
We make no such attempt for the time being,
instead passing our SAT proofs through the
formally verified proof checker \texttt{cake\_lpr}~\cite{tanVerifiedPropagationRedundancy2023}
and asserting unsatisfiability of the CNF as an axiom in Lean.
Thus we trust that the CNF formula produced by the verified Lean encoder
is the same one whose unsatisfiability was checked by \textsf{cake\_lpr}.

\section{Concluding Remarks}\label{sec:conclusions}
We have proved the correctness of the main result of Heule and Scheucher~\cite{emptyHexagonNumber},
implying $h(6) \leq 30$.
Given that the lower bound $h(6) > 29$ can be checked directly (see \cite{emptyHexagonNumber}),
we conclude the result $h(6) = 30$ is indeed correct.
We believe this work puts a \emph{happy ending} to
one line of research started by Klein, Erd\H{o}s and Szekeres in the 1930s.
Prior to formalization, the result of Heule and Scheucher
relied on the correctness of various components of a highly sophisticated encoding
that are hard to validate manually.
We developed a significant theory of combinatorial geometry
that was not present in~\textsf{mathlib}.
Beyond the main theorem presented here,
we showed how our framework can be used for other related theorems
such as $g(6) = 17$,
and we hope it can be used for proving many further results in the area.

Our formalization required approximately 300 hours of work over 3 months
by researchers with significant experience formalizing mathematics in Lean.
The final version of our proofs consists of approximately 4.7k lines of Lean code;
% (including comments/whitespace);
about $26\%$ are lemmas that should be moved to upstream libraries,
about $40\%$ develops the theory of orientations in plane geometry,
and the remaining $34\%$ (1550 LOC) validates the symmetry breaking and SAT encoding.
% Many more lines of code were deleted or rewritten,
% so these numbers should be taken with a grain of salt.
%  (BS: commented to save space)

We substantially simplified the symmetry-breaking argument presented by Heule and Scheucher,
and derived in turn from Scheucher~\cite{scheucherTwoDisjoint5holes2020}.
Moreover, we found a small error in their proof,
as their transformation uses the mapping $(x, y) \mapsto (x/y, -1/y)$,
and incorrectly assumes that $x/y$ is increasing for points in CCW-order,
whereas only the slopes $y/x$ are increasing.
Similarly, we found a typo in the statement of the \textsf{Lex order} condition
that did not match the (correct) code of Heule and Scheucher.
Our formalization corrects this.

In terms of future work,
we hope to formally verify the result $h(7) = \infty$ due to Horton~\cite{hortonSetsNoEmpty1983},
and other results in Erd\H{o}s-Szekeres style problems.
A key challenge for the community
is to improve the connection between verified SAT tools and ITPs.
This presents a significant engineering task
for proofs that are hundreds of terabytes long (as in this result).
Although we are confident that our results are correct,
the trust story at this connection point has room for improvement.

\bibliography{main}

\begin{thebibliography}{10}

\bibitem{biereHandbookSatisfiabilityVolume2009}
A.~Biere, M.~Heule, H.~{van Maaren}, and T.~Walsh.
\newblock {\em Handbook of {{Satisfiability}}: {{Volume}} 185 {{Frontiers}} in
  {{Artificial Intelligence}} and {{Applications}}}.
\newblock {IOS Press}, {NLD}, 2009.

\bibitem{brakensiek2023resolution}
Joshua Brakensiek, Marijn Heule, John Mackey, and David Narváez.
\newblock {The Resolution of Keller's Conjecture}, 2023.
\newblock \href {http://arxiv.org/abs/1910.03740} {\path{arXiv:1910.03740}}.

\bibitem{21bright_sat_based_resolution_lams_problem}
Curtis Bright, Kevin K.~H. Cheung, Brett Stevens, Ilias~S. Kotsireas, and Vijay
  Ganesh.
\newblock A {SAT}-based resolution of {L}am's {P}roblem.
\newblock In {\em Thirty-Fifth {AAAI} Conference on Artificial Intelligence,
  {AAAI} 2021}, pages 3669--3676. {AAAI} Press, 2021.
\newblock URL: \url{https://doi.org/10.1609/aaai.v35i5.16483}, \href
  {https://doi.org/10.1609/AAAI.V35I5.16483}
  {\path{doi:10.1609/AAAI.V35I5.16483}}.

\bibitem{cluster}
Shawn~T. Brown, Paola Buitrago, Edward Hanna, Sergiu Sanielevici, Robin Scibek,
  and Nicholas~A. Nystrom.
\newblock {\em Bridges-2: A Platform for Rapidly-Evolving and Data Intensive
  Research}, pages 1--4.
\newblock Association for Computing Machinery, New York, NY, USA, 2021.

\bibitem{Castelvecchi2021}
Davide Castelvecchi.
\newblock Mathematicians welcome computer-assisted proof in 'grand unification'
  theory.
\newblock {\em Nature}, 595(7865):18–19, June 2021.
\newblock URL: \url{http://dx.doi.org/10.1038/d41586-021-01627-2}, \href
  {https://doi.org/10.1038/d41586-021-01627-2}
  {\path{doi:10.1038/d41586-021-01627-2}}.

\bibitem{Cayden}
Cayden Codel, Marijn J.~H. Heule, and Jeremy Avigad.
\newblock {Verified Encodings for SAT Solvers}.
\newblock In Alexander Nadel and Kristin~Yvonne Rozier, editors, {\em
  Proceedings of the 23rd conference on Formal Methods In Computer-Aided
  Design}, 2023.

\bibitem{Crawford}
James Crawford, Matthew Ginsberg, Eugene Luks, and Amitabha Roy.
\newblock Symmetry-breaking predicates for search problems.
\newblock In {\em Proc. KR'96, 5th Int. Conf. on Knowledge Representation and
  Reasoning}, pages 148--159. Morgan Kaufmann, 1996.

\bibitem{formalPythagoreanTriples}
Lu\'{i}s Cruz-Filipe, Jo\~{a}o Marques-Silva, and Peter Schneider-Kamp.
\newblock {Formally Verifying the Solution to the Boolean Pythagorean Triples
  Problem}.
\newblock {\em J. Autom. Reason.}, 63(3):695–722, oct 2019.
\newblock \href {https://doi.org/10.1007/s10817-018-9490-4}
  {\path{doi:10.1007/s10817-018-9490-4}}.

\bibitem{LPAR-21:Formally_Proving_Boolean_Pythagorean}
Lu\'{i}s Cruz-Filipe and Peter Schneider-Kamp.
\newblock {Formally Proving the Boolean Pythagorean Triples Conjecture}.
\newblock In Thomas Eiter and David Sands, editors, {\em LPAR-21. 21st
  International Conference on Logic for Programming, Artificial Intelligence
  and Reasoning}, volume~46 of {\em EPiC Series in Computing}, pages 509--522.
  EasyChair, 2017.
\newblock URL: \url{https://easychair.org/publications/paper/xq6J}, \href
  {https://doi.org/10.29007/jvdj} {\path{doi:10.29007/jvdj}}.

\bibitem{demouraLeanTheoremProver2015}
Leonardo {de Moura}, Soonho Kong, Jeremy Avigad, Floris {van Doorn}, and Jakob
  {von Raumer}.
\newblock The {{Lean Theorem Prover}} ({{System Description}}).
\newblock In Amy~P. Felty and Aart Middeldorp, editors, {\em Automated
  {{Deduction}} - {{CADE-25}}}, pages 378--388, {Cham}, 2015. {Springer
  International Publishing}.

\bibitem{60erdos_some_extremum_problems_elementary_geometry}
Paul Erd{\H{o}}s and George Szekeres.
\newblock On some extremum problems in elementary geometry.
\newblock {\em Ann. Univ. Sci. Budapest. E{\"o}tv{\"o}s Sect. Math.},
  3(4):53--62, 1960.

\bibitem{35erdos_combinatorial_problem_geometry}
Paul Erd{\H{o}}s and Gy{\"o}rgy Szekeres.
\newblock A combinatorial problem in geometry.
\newblock {\em Compositio Mathematica}, 2:463--470, 1935.
\newblock URL: \url{http://eudml.org/doc/88611}.

\bibitem{felsnerSweepsArrangementsSignotopes2001}
Stefan Felsner and Helmut Weil.
\newblock Sweeps, arrangements and signotopes.
\newblock {\em Discrete Applied Mathematics}, 109(1):67--94, April 2001.
\newblock \href {https://doi.org/10.1016/S0166-218X(00)00232-8}
  {\path{doi:10.1016/S0166-218X(00)00232-8}}.

\bibitem{gerkenEmptyConvexHexagons2008}
Tobias Gerken.
\newblock {Empty Convex Hexagons in Planar Point Sets}.
\newblock {\em Discrete \& Computational Geometry}, 39(1):239--272, mar 2008.
\newblock \href {https://doi.org/10.1007/s00454-007-9018-x}
  {\path{doi:10.1007/s00454-007-9018-x}}.

\bibitem{GilAndWennerbeck}
Sofia Giljeg{\r{a}}rd and Johan Wennerbeck.
\newblock {Puzzle Solving with Proof}.
\newblock Master's thesis, Chalmers University of Technology, 2021.

\bibitem{gowers2023conjecture}
W.~T. Gowers, Ben Green, Freddie Manners, and Terence Tao.
\newblock On a conjecture of {M}arton, 2023.
\newblock \href {http://arxiv.org/abs/2311.05762} {\path{arXiv:2311.05762}}.

\bibitem{Harborth1978}
Heiko Harborth.
\newblock {Konvexe Fünfecke in ebenen Punktmengen.}
\newblock {\em Elemente der Mathematik}, 33:116--118, 1978.
\newblock URL: \url{http://eudml.org/doc/141217}.

\bibitem{Heule_2016}
Marijn J.~H. Heule, Oliver Kullmann, and Victor~W. Marek.
\newblock {\em {Solving and Verifying the Boolean Pythagorean Triples Problem
  via Cube-and-Conquer}}, page 228–245.
\newblock Springer International Publishing, 2016.
\newblock URL: \url{http://dx.doi.org/10.1007/978-3-319-40970-2_15}, \href
  {https://doi.org/10.1007/978-3-319-40970-2_15}
  {\path{doi:10.1007/978-3-319-40970-2_15}}.

\bibitem{emptyHexagonNumber}
Marijn J.~H. Heule and Manfred Scheucher.
\newblock Happy ending: An empty hexagon in every set of 30 points, 2024.
\newblock \href {http://arxiv.org/abs/2403.00737} {\path{arXiv:2403.00737}}.

\bibitem{holmsen2017two}
Andreas~F Holmsen, Hossein~Nassajian Mojarrad, J{\'a}nos Pach, and G{\'a}bor
  Tardos.
\newblock Two extensions of the erd{\H{o}}s-szekeres problem.
\newblock {\em arXiv preprint arXiv:1710.11415}, 2017.

\bibitem{hortonSetsNoEmpty1983}
J.~D. Horton.
\newblock {Sets with No Empty Convex 7-Gons}.
\newblock {\em Canadian Mathematical Bulletin}, 26(4):482--484, 1983.
\newblock \href {https://doi.org/10.4153/CMB-1983-077-8}
  {\path{doi:10.4153/CMB-1983-077-8}}.

\bibitem{knuthAxiomsHulls1992}
Donald~E. Knuth.
\newblock Axioms and {{Hulls}}.
\newblock In Donald~E. Knuth, editor, {\em Axioms and {{Hulls}}}, Lecture
  {{Notes}} in {{Computer Science}}, pages 1--98. {Springer}, {Berlin,
  Heidelberg}, 1992.
\newblock \href {https://doi.org/10.1007/3-540-55611-7_1}
  {\path{doi:10.1007/3-540-55611-7_1}}.

\bibitem{konev2014sat}
Boris Konev and Alexei Lisitsa.
\newblock {A SAT Attack on the Erdos Discrepancy Conjecture}, 2014.
\newblock \href {http://arxiv.org/abs/1402.2184} {\path{arXiv:1402.2184}}.

\bibitem{lammichEfficientVerifiedSAT2020}
Peter Lammich.
\newblock {Efficient Verified (UN)SAT Certificate Checking}.
\newblock {\em Journal of Automated Reasoning}, 64(3):513--532, March 2020.
\newblock \href {https://doi.org/10.1007/s10817-019-09525-z}
  {\path{doi:10.1007/s10817-019-09525-z}}.

\bibitem{10maric_formal_verification_modern_sat_solver_shallow_embedding_isabelle_hol}
Filip Maric.
\newblock Formal verification of a modern {SAT} solver by shallow embedding
  into {I}sabelle/{HOL}.
\newblock {\em Theor. Comput. Sci.}, 411(50):4333--4356, 2010.
\newblock URL: \url{https://doi.org/10.1016/j.tcs.2010.09.014}, \href
  {https://doi.org/10.1016/J.TCS.2010.09.014}
  {\path{doi:10.1016/J.TCS.2010.09.014}}.

\bibitem{19maric_fast_formal_proof_erdos_szekeres_conjecture_convex_polygons_most_six_points}
Filip Maric.
\newblock Fast formal proof of the {E}rd{\H{o}}s-{S}zekeres conjecture for
  convex polygons with at most 6 points.
\newblock {\em J. Autom. Reason.}, 62(3):301--329, 2019.
\newblock URL: \url{https://doi.org/10.1007/s10817-017-9423-7}, \href
  {https://doi.org/10.1007/S10817-017-9423-7}
  {\path{doi:10.1007/S10817-017-9423-7}}.

\bibitem{The_mathlib_Community_2020}
The mathlib Community.
\newblock The {L}ean mathematical library.
\newblock In {\em Proceedings of the 9th ACM SIGPLAN International Conference
  on Certified Programs and Proofs}, POPL ’20. ACM, January 2020.
\newblock URL: \url{http://dx.doi.org/10.1145/3372885.3373824}, \href
  {https://doi.org/10.1145/3372885.3373824}
  {\path{doi:10.1145/3372885.3373824}}.

\bibitem{nicolasEmptyHexagonTheorem2007}
Carlos~M. Nicolas.
\newblock {The Empty Hexagon Theorem}.
\newblock {\em Discrete \& Computational Geometry}, 38(2):389--397, September
  2007.
\newblock \href {https://doi.org/10.1007/s00454-007-1343-6}
  {\path{doi:10.1007/s00454-007-1343-6}}.

\bibitem{oeVersatVerifiedModern2012}
Duckki Oe, Aaron Stump, Corey Oliver, and Kevin Clancy.
\newblock {Versat: A Verified Modern SAT Solver}.
\newblock In Viktor Kuncak and Andrey Rybalchenko, editors, {\em {Verification,
  Model Checking, and Abstract Interpretation}}, pages 363--378, {Berlin,
  Heidelberg}, 2012. {Springer Berlin Heidelberg}.

\bibitem{scheucherTwoDisjoint5holes2020}
Manfred Scheucher.
\newblock Two disjoint 5-holes in point sets.
\newblock {\em Computational Geometry}, 91:101670, December 2020.
\newblock \href {https://doi.org/10.1016/j.comgeo.2020.101670}
  {\path{doi:10.1016/j.comgeo.2020.101670}}.

\bibitem{skotam_creusat_2022}
Sarek~Høverstad Skotåm.
\newblock {CreuSAT, Using Rust and Creusot to create the world's fastest
  deductively verified SAT solver}.
\newblock Master's thesis, University of Oslo, 2022.
\newblock URL: \url{https://www.duo.uio.no/handle/10852/96757}.

\bibitem{slomanATeamMathProves2023}
Leila Sloman.
\newblock {`A-Team' of Math Proves a Critical Link Between Addition and Sets}.
\newblock
  https://www.quantamagazine.org/a-team-of-math-proves-a-critical-link-between-addition-and-sets-20231206/,
  December 2023.

\bibitem{Subercaseaux_Heule_2023}
Bernardo Subercaseaux and Marijn J.~H. Heule.
\newblock {The Packing Chromatic Number of the Infinite Square Grid is 15}.
\newblock In Sriram Sankaranarayanan and Natasha Sharygina, editors, {\em Tools
  and Algorithms for the Construction and Analysis of Systems - 29th
  International Conference, TACAS 2023, Held as Part of ETAPS 2022,
  Proceedings, Part I}, volume 13993 of {\em Lecture Notes in Computer
  Science}, page 389–406. Springer, 2023.
\newblock \href {https://doi.org/10.1007/978-3-031-30823-9_20}
  {\path{doi:10.1007/978-3-031-30823-9_20}}.

\bibitem{subercaseaux2023minimizing}
Bernardo Subercaseaux, John Mackey, Marijn J.~H. Heule, and Ruben Martins.
\newblock Minimizing pentagons in the plane through automated reasoning, 2023.
\newblock \href {http://arxiv.org/abs/2311.03645} {\path{arXiv:2311.03645}}.

\bibitem{suk2017erdos}
Andrew Suk.
\newblock On the erd{\H{o}}s-szekeres convex polygon problem.
\newblock {\em Journal of the American Mathematical Society}, 30(4):1047--1053,
  2017.

\bibitem{szekeres_peters_2006}
George Szekeres and Lindsay Peters.
\newblock Computer solution to the 17-point {Erd\H{o}s-Szekeres} problem.
\newblock {\em The ANZIAM Journal}, 48(2):151--164, 2006.
\newblock \href {https://doi.org/10.1017/S144618110000300X}
  {\path{doi:10.1017/S144618110000300X}}.

\bibitem{06szekeres_computer_solution_17_point_erdos_szekeres_problem}
George Szekeres and Lindsay Peters.
\newblock Computer solution to the 17-point erd{\H{o}}s-szekeres problem.
\newblock {\em The ANZIAM Journal}, 48(2):151--164, 2006.

\bibitem{tanVerifiedPropagationRedundancy2023}
Yong~Kiam Tan, Marijn J.~H. Heule, and Magnus~O. Myreen.
\newblock {Verified Propagation Redundancy and Compositional UNSAT Checking in
  CakeML}.
\newblock {\em International Journal on Software Tools for Technology
  Transfer}, 25(2):167--184, April 2023.
\newblock \href {https://doi.org/10.1007/s10009-022-00690-y}
  {\path{doi:10.1007/s10009-022-00690-y}}.

\bibitem{Walters2004ItAT}
Mark Walters.
\newblock {It Appears That Four Colors Suffice : A Historical Overview of the
  Four-Color Theorem}.
\newblock 2004.
\newblock URL: \url{https://api.semanticscholar.org/CorpusID:14382286}.

\bibitem{drat-trim14}
Nathan Wetzler, Marijn J.~H. Heule, and Warren~A. Hunt.
\newblock {DRAT}-trim: Efficient checking and trimming using expressive clausal
  proofs.
\newblock In Carsten Sinz and Uwe Egly, editors, {\em Theory and Applications
  of Satisfiability Testing -- SAT 2014}, pages 422--429, Cham, 2014. Springer
  International Publishing.

\end{thebibliography}

\end{document}